%% file: draft8.tex
\documentclass[%
 aip,
 amsmath,amssymb,
reprint, onecolumn 
]{revtex4-1}
\usepackage[caption=false]{subfig}

\usepackage{graphicx}
\usepackage{dcolumn}
\usepackage{bm}

\usepackage[utf8]{inputenc}
\usepackage[T1]{fontenc}
\usepackage{mathptmx}
\usepackage{etoolbox}
\usepackage{amsthm}
\usepackage{physics}
\usepackage{color}
\usepackage{enumitem}
\usepackage{float}
\usepackage{bbm}
\usepackage{pgf}

\graphicspath{ {./figures/} }

\theoremstyle{plain}

\newtheorem{definition}{Definition}
\newtheorem{fact}{Fact}
\newtheorem{remark}{Remark}
\newtheorem{theorem}{Theorem}
\newtheorem{lemma}{Lemma}
\newtheorem{corollary}[theorem]{Corollary}

\DeclareMathOperator{\sign}{sign}

\def\id{\mathrm{id}}
\def\d{\text{d}}
\newcommand{\dist}[1]{\mathrm{dist}(#1)}
\newcommand{\unorm}[1]{\lvert\lvert\lvert #1 \rvert\rvert\rvert}

\makeatletter
\def\@email#1#2{%
 \endgroup
 \patchcmd{\titleblock@produce}
  {\frontmatter@RRAPformat}
  {\frontmatter@RRAPformat{\produce@RRAP{*#1\href{mailto:#2}{#2}}}\frontmatter@RRAPformat}
  {}{}
}%
\makeatother
\begin{document}

\preprint{AIP/123-QED}

\title[Time-optimal quantum dynamics under fixed-norm generators]{Time-optimal quantum dynamics under fixed-norm generators}

\author{Oskar A. Prośniak}%
\email{oskar.prosniak@uni.lu}
\affiliation{Department of Physics and Materials Science,
University of Luxembourg,\\
Luxembourg, G.~D.~Luxembourg}

\author{Karol Życzkowski}
\affiliation{Institute of Theoretical Physics, Jagiellonian University, Kraków, Poland
}
\affiliation{Center for Theoretical Physics, Polish Academy of Sciences, Warsaw, Poland}

\date{October 30, 2026}

\begin{abstract}
We investigate the dynamics of quantum systems whose evolution is governed by linear generators
of fixed unitarily invariant norm. For this broad class of evolutions, we derive a universal time-optimal generator of the dynamics converting one given state into another. It provides a time-optimal non-Hermitian Hamiltonian for pure state conversion. We then use this general result to derive quantum speed limits for mixed-state evolutions governed by Lindbladians and non-Hermitian Hamiltonians.
Our findings follow from the geometric properties of the group $GL(n,\mathbb{C})$, in particular the geodesic distance. The derived bounds hold for an arbitrary system size $n$, depend only on the norms of the evolution generator and the initial and final states, and do not require performing any time-averaging.
\end{abstract}

\maketitle
\section{Introduction}
The task of converting one quantum state into another in the shortest possible time lies at the very heart of now rapidly developing quantum technologies. Fast state transformation plays an important role in quantum computing\cite{suzuki_performance_2020}, quantum communication\cite{murphy_communication_2010, deffner_quantum_2020}, quantum state preparation\cite{bukov_geometric_2019}, quantum complexity\cite{schulte-herbruggen_optimal_2005}, and thermometry\cite{campbell_precision_2018}. From the perspective of potential applications it is especially important to characterise time-optimal evolutions in open systems, which not only model the unavoidable influence of the environment, but may also allow for speed-ups with respect to closed system protocols.

Two complementary approaches to this problem have emerged in the literature. First are the so called quantum speed limits (QSLs), which are lower bounds on time required for the state conversion under a particular class of evolutions\cite{deffner_quantum_2017}. Their earliest examples are the Mandelstam-Tamm time-energy uncertainty relation\cite{mandelstam_uncertainty_1991} and the Margolus-Levitin bound\cite{margolus_maximum_1998}. They both concern unitary evolution of pure states and were later confirmed experimentally\cite{ness_observing_2021,lam_demonstration_2021, wu_testing_2024}, generalized to time-dependent Hamiltonians\cite{poggi_quantum_2013, hornedal_closed_2023}, and to unitary dynamics of mixed states\cite{giovannetti_quantum_2003, campaioli_tightening_2018}. Further extensions to open system dynamics followed and completely new bounds were developed. These utilise different notions of distinguishability of quantum states (metrics), such as quantum Fisher information\cite{taddei_quantum_2013}, relative purity\cite{delcampo_quantum_2013}, alternative fidelity\cite{mai_tight_2023}, trace distance\cite{nakajima_speed_2022, funo_speed_2019}, Hilbert-Schmidt distance\cite{campaioli_tight_2019}, and quantum entropy\cite{pires_unified_2022}.

\begin{figure}[b]
    \centering
    \includegraphics[width=0.45\linewidth]{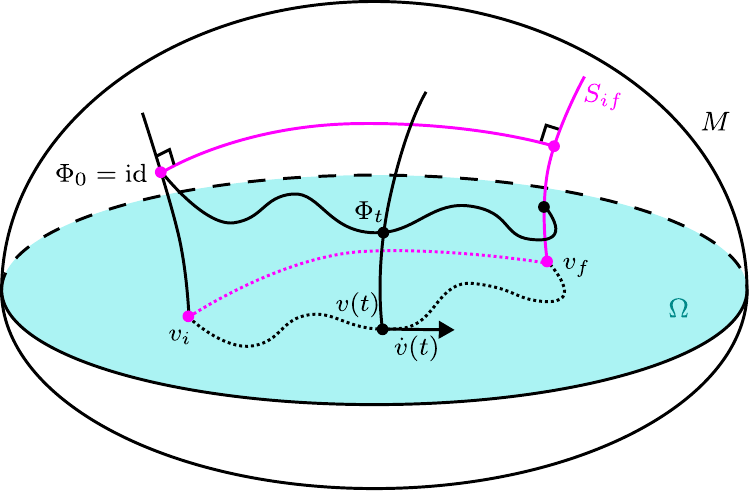}
    \caption{An illustration of the bundle construction utilised in the main lemma. Here, $v_i, v_f\in \Omega$ stand for the initial and final states of the dynamics. These can be either pure (unnormalised) quantum states, in which case the base manifold $\Omega$ is simply $\mathbb{C}^n\setminus\{0\}$, or mixed - then $\Omega = \mathbb{C}^{n\times n}\setminus\{0\}$. We take the total space $M$ to be the space of all automorphisms $\Phi$ of the underlying base manifold. In the case of pure states, they correspond to maps generated by non-Hermitian Hamiltonians, in the case of mixed states they include, but are not restricted to, Lindblad and non-Hermitian Hamiltonian generated maps. We construct fibres to consist of maps sending the corresponding base point to $v_i$ and otherwise unrestricted. In particular, the fibre $S_{if}$ sends $v_i$ to $v_f$.
    The distance between $v_i$ and $v_f$ equals the length of the shortest geodesic between $\Phi_0=\mathrm{id}$ and $S_{if}$, which must be orthogonal to all the fibres. For a natural choice of metric on $M$ the speed $\norm{\dot{v}(t)}$ of any path $t \mapsto v(t)$ is locally upper-bounded by any unitarily invariant norm of the corresponding generator, on which we impose restrictions.
    }
    \label{fig: cartoon}
\end{figure}

The second approach is the quantum time-optimal control, including quantum brachistochrone problems\cite{carlini_timeoptimal_2006, carlini_time_2008, caneva_optimal_2009, wang_quantum_2015, lin_timeoptimal_2020, garcia_highly_2022, koike_quantum_2022}. It aims at finding the exact form of the fastest evolution within the class of considered dynamics. Typically, the admissible generators depend on a set of external control fields subject to amplitude constraints. The setting of the problem can then exactly mirror the experimental reality, at the price of highly specialised results. Analysis of such problems often utilises variational approaches similar to the one known from the classical brachistochrone problem. While these methods are well suited for numerical optimisation, exact analytical results are scarce and limited to the simplest cases only\cite{ansel_introduction_2024}. The works concerning the non-Hermitian quantum systems with bounded spectral gap are here a notable exception\cite{bender_faster_2007,nesterov_optimum_2009}.

Both these approaches have their internal limitations. QSLs, while identifying the fundamental constraints of quantum evolution, usually do not characterise the dynamics saturating them. On the other hand, the results in time-optimal control are typically tailored to very specific systems and are hence non-universal. 
One of the promising ways of bridging these two perspectives is to utilise the geometrical techniques that already provided ground for a clear interpretation and development of several QSLs\cite{pires_generalized_2016}.

Such geometric approaches typically consist of endowing the space of accessible quantum states with a differential geometric structure that provides compatible measures of speed and distance. Using triangle inequality one then formulates a QSL as a ratio between the distance and the average speed of evolution between the two states of interest. If the average evolution speed is upper bounded in the (possibly further restricted) class of evolutions of interest, this procedure provides an operationally meaningful bound. 
However, if the space of accessible quantum states does not admit a differentiable manifold structure, the direct application of described procedure becomes cumbersome if not impossible. In particular, the set of all (mixed) states of a given dimension generally does not form a differentiable manifold\cite{dandrea_pseudomanifold_2021}, while it is accessible by a (time-dependent, possibly unbounded) Lindblad dynamics\cite{lindblad_generators_1976,gorini_completely_1976,chruscinski_brief_2017,cai_how_2025a}. Additionally, most geometric QSLs are often not straightforward to apply in practice, because of the incompatibility of the experimental limitations and the notion of speed used. It is especially difficult when the instantaneous speed is state dependent, which is most often the case.

Our main result for the time-optimal dynamics governed by generators of fixed unitarily-invariant norm (e.g.\ operator, trace or Hilbert-Schmidt) avoids the aforementioned problems. The norm restriction provides a viable experimental `knob', while additionally arising naturally from the geometric structure of the space of invertible linear maps. We also prove two weaker statements, in the form of QSLs, for two sub-classes of such dynamics, namely Lindblad and non-Hermitian evolutions of mixed states. In these cases, we are able to apply the geometrical reasoning, at the price of extending the space to non-physical maps and states, and in general rendering our QSLs not tight.

\section{Setting and results}
\subsection{Setting and definitions}
At any $n$-dimensional Hilbert space $\mathcal{H}$ we define:
\begin{definition}[Unitarily invariant norm] 
    Let $\mathcal{B}(\mathcal{H})$ ($\mathcal{T}(\mathcal{H}):=\mathcal{B}(\mathcal{B}(\mathcal{H}))$) be a space of all bounded (super-)operators. We define the unitarily invariant norm $\unorm{\bullet}: \mathcal{B}(\mathcal{H}) \to \mathbb{R}_{\geq 0}$ ($\unorm{\bullet}: \mathcal{T}(\mathcal{H}) \to \mathbb{R}_{\geq 0}$) to be a norm satisfying
    \begin{enumerate}
        \item $\unorm{U \bullet V} = \unorm{\bullet}$, for all unitary (super-) operators $U,V\in \mathcal{B}(\mathcal{H})$ ($U,V\in \mathcal{T}(\mathcal{H})$),
        \item $\unorm{P} = 1$, for all rank-$1$ partial isometries $P\in\mathcal{B}(\mathcal{H})$ ($P\in \mathcal{T}(\mathcal{H})$).
    \end{enumerate}
\end{definition}

As a particular example of unitarily invariant norms we consider the so called Schatten $p$-norms.

\begin{definition}[Schatten $p$-norm]
    Let $\lbrace E_i\rbrace_{i=1}^{n^2}$ be an orthonormal (wrt the Hilbert-Schmidt inner product) basis of $\mathcal{B}(\mathcal{H})$, and let $p\in [1,\infty]$. For any super-operator $T\in \mathcal{T}(\mathcal{H})$ by its matrix representation we mean $T_{ij}:=\Tr[E_i^\dagger T(E_j)]$. The Schatten $p$-norm of $T$ is then defined as
    \begin{equation}
        \norm{T}_p := \left(\Tr[\abs{(T_{ij})_{i,j=1}^{n^2}}^p]\right)^{1/p}.
    \end{equation}
    Analogously, for any operator $O\in \mathcal{B}(\mathcal{H})$, its Schatten $p$-norm is defined as
    \begin{equation}
        \norm{O}_p := \left(\Tr[\abs{O}^p]\right)^{1/p}.
    \end{equation}
\end{definition}

Note that we do not mean here a norm induced by a Schatten norm on the space of operators, but rather the Schatten norm of a super-operator matrix. The Schatten $p$-norms include the familiar trace norm $\norm{\bullet}_*:=\norm{\bullet}_1$, the Hilbert-Schmidt norm $\norm{\bullet}_{HS}:=\norm{\bullet}_2$, and the operator norm $\norm{\bullet}_{op}:=\norm{\bullet}_\infty$. Importantly, the operator norm is the smallest of the unitarily invariant norms in the following sense.

 \begin{fact}\label{fact: operator norm}
     For any (super-)operator $O\in \mathcal{B}(\mathcal{H})$ ($T\in\mathcal{T}(\mathcal{H})$) and any unitarily invariant norm $\unorm{\bullet}$ it holds that\cite{bhatia_matrix_1997}
 \begin{equation}
     \norm{O}_{op} \leq \unorm{O} \quad (\norm{T}_{op} \leq \unorm{T}).
 \end{equation}
 \end{fact}

\begin{definition}[Quantum states]
    We define the sets of
    \begin{enumerate}[label = \alph*)]
        \item un-normalised quantum states
            \begin{equation}
            \widetilde{\mathcal{D}}(\mathcal{H}) := \left\lbrace \rho \in \mathcal{B}(\mathcal{H}) | \rho = \rho^\dagger, \rho \geq 0,\rho \not\equiv 0\right\rbrace,
            \end{equation}
        \item normalised quantum states
            \begin{equation}
            \mathcal{D}(\mathcal{H}) := \left\lbrace \rho \in \mathcal{B}(\mathcal{H}) | \rho=\rho^\dagger, \rho \geq 0, \Tr[\rho]=1\right\rbrace.
            \end{equation}
    \end{enumerate}
    
\end{definition}

\begin{definition}[Lindbladian]
    Let $\{ F_i \}_{i=1}^{n^2-1}$ be an orthonormal set of traceless operators on $\mathcal{H}$ wrt the inner product induced by the Hilbert-Schmidt norm. For any choice of bounded, Hermitian Hamiltonian $H\in \mathcal{B}(\mathcal{H})$ and the Kossakowski matrix $K \in \mathbb{C}^{(n^2-1)\times (n^2-1)}$, $K\geq 0$, we call the following superoperator a Lindblad generator\cite{lindblad_generators_1976,gorini_completely_1976}
    \begin{equation}
        L: \mathcal{B}(\mathcal{H}) \to \mathcal{B}(\mathcal{H}),\quad L(\rho) := -i[H,\rho] + \sum_{i,j}^{n^2-1} K_{ij} \left[ F_i \rho F_j^\dagger - \frac{1}{2}\left( F_j^\dagger F_i \rho +\rho F_j^\dagger F_i\right)\right].
    \end{equation}
\end{definition}

In this work, we are interested in time-optimal quantum evolutions governed by linear generators, such as Lindblad generators or non-Hermitian Hamiltonians, with a fixed unitarily-invariant norm. The following section states this problem more rigorously.

\subsection{Main results}
The main result of this work is the following general theorem. We use here a standard convention $\inf \emptyset = \infty$ and natural units $\hbar = 1$. The norm induced by the standard inner product on the Hilbert space will be denoted by $\norm{\bullet}$ without subscripts.

\begin{theorem}[Time-optimal evolution of pure states]\label{thm: speed limit pure}
     Let $\mathcal{H}$ be an $n$-dimensional Hilbert space, $\unorm{\bullet}: \mathcal{B}(\mathcal{H}) \to \mathbb{R}_{\geq 0}$ a unitarily invariant norm, and let $\ket{\psi_i},\ket*{\psi_f} \in \mathcal{H}$ be two non-zero pure quantum states (not necessarily normalised). Define
    \begin{align}
        \mathfrak{G}(\mathcal{H}) &:= \left\lbrace G \in \mathrm{L}^1\!\left([0,\infty),\mathcal{B}(\mathcal{H})\right)|\, \forall t\geq 0:\,\unorm{G(t)}=1\right\rbrace,\\
        \tau_{if}^\text{pure} &:= \inf \left\{t\geq 0 \,| \,  \exists G\in \mathfrak{G}(\mathcal{H}): \mathbf{T}e^{-i\int_0^t G(s)\,\d s}\ket{\psi_i} = \ket*{\psi_f}\right\},
    \end{align}
    where $\mathbf{T}$ stands for the time ordering operator. Denote $\lambda:=\log\left(\frac{\norm*{\psi_f}}{\norm{\psi_i}}\right)$, $\theta:=\arccos\left(\frac{\Re \braket{\psi_i}{\psi_f}}{\norm{\psi_i}\norm*{\psi_f}}\right)$. The shortest time of pure state evolution generated by a normalised non-Hermitian Hamiltonian equals then 
    \begin{equation}\label{eq: limit pure}
         \tau_{if}^\text{pure} = \sqrt{ \lambda^2 + \theta^2 },
    \end{equation}
    and is attained by
    \begin{equation}
        G^*(t) := i \frac{\ketbra{\dot{\psi}(t)}{\psi(t)}}{\braket{\psi(t)}},\quad \ket{\psi(t)}:=
        \begin{cases}
            \norm{\psi_i} e^{t \lambda /\sqrt{\lambda^2+\theta^2}}\left( \frac{\sin(\theta-\frac{t\theta}{\sqrt{\lambda^2+\theta^2}})}{\sin(\theta)} \frac{\ket{\psi_i}}{\norm{\psi_i}} + \frac{\sin(\frac{t \theta}{\sqrt{\lambda^2+\theta^2}})}{\sin(\theta)} \frac{\ket{\psi_f}}{\norm{\psi_f}} \right), &\theta\in (0,\pi),\\
            \norm{\psi_i}e^{t \lambda /\sqrt{\lambda^2+\theta^2}}\left( \cos(\frac{t \theta}{\sqrt{\lambda^2+\theta^2}}) \frac{\ket{\psi_i}}{\norm{\psi_i}} + \sin(\frac{t \theta}{\sqrt{\lambda^2+\theta^2}}) \ket{\psi_i^\perp} \right), &\theta = \pi, \\
            e^{t \lambda /\sqrt{\lambda^2+\theta^2}} \ket{\psi_i}, &\theta=0,\\
        \end{cases}
    \end{equation}
    where $\ket{\psi_i^\perp}$ is an arbitrary normalised state satisfying $\Re \braket{\psi_i}{\psi_i^\perp}=0$ (the solution is non-unique).
\end{theorem} 

Two observations are here in order. First, the result is independent of the choice of the unitarily invariant norm. Second, for $\theta=0$ the optimal non-Hermitian Hamiltonian is time independent and purely anti-Hermitian.

Two useful corollaries, which are particularly well suited for practical applications, follow from this general theorem.

\begin{corollary}[Speed limit for Lindblad dynamics]\label{cor: speed limit Lindblad}
    Let $\mathcal{H}$ be an $n$-dimensional Hilbert space, $\unorm{\bullet}: \mathcal{T}(\mathcal{H}) \to \mathbb{R}_{\geq 0}$ a unitarily invariant norm, and let $\rho_i,\rho_f \in \mathcal{D}(\mathcal{H})$ be two arbitrary quantum states. Define
    \begin{align}
        \mathfrak{L}(\mathcal{H}) &:= \left\lbrace L \in \mathcal{T}(\mathcal{H})|\, L \text{ is a Lindblad generator}, \unorm{L}=1\right\rbrace,\\
        \tau_{if}^{Lind} &:= \inf \left\{t\geq 0 \,| \,  \exists L\in \mathfrak{L}(\mathcal{H}): e^{Lt}\rho_i = \rho_f\right\}.
    \end{align}
    Then the time of evolution under a normalised Lindblad generator $L$ satisfies 
    \begin{equation}\label{eq: limit Lindblad}
         \tau_{if}^{Lind}\geq \sqrt{\log^2
         \left(\sqrt{
         \frac{\Tr\big[\rho_f^2\big]}{\Tr[\rho_i^2]}
         }\right) + \arccos^2
         \left(
         \frac{ \Tr[\rho_i \rho_f]}{\sqrt{\Tr[\rho_i^2]\Tr\big[\rho_f^2\big]}}
         \right)}.
    \end{equation}
\end{corollary}

The above bound is actually satisfied by a much larger class of evolution with linear time-dependent generators of fixed norm, as can be seen from Theorem~\ref{thm: speed limit pure}. We pick the Lindblad dynamics as an illustrative example. For instance, this Corollary can be easily extended to quantum evolutions governed by non-Hermitian Hamiltonians whose associated Liouvillians have fixed norm. It is however more practical to bound the norm of the Hamiltonian itself. The following Corollary gives a QSL in such a scenario.

\begin{corollary}[Speed limit for non-Hermitian dynamics]\label{cor: speed limit non-Herm}
     Let $\mathcal{H}$ be an $n$-dimensional Hilbert space, and let $\rho_i,\rho_f \in \widetilde{\mathcal{D}}(\mathcal{H})$ be two arbitrary quantum states. Let $\mathfrak{G}(\mathcal{H})$ be as in Theorem~\ref{thm: speed limit pure}, and define
    \begin{equation}
        \tau_{if}^{NH} := \inf \left\{t\geq 0 \,| \,  \exists G\in \mathfrak{G}(\mathcal{H}): \mathbf{T}e^{-i\int_0^t G(s)\,\d s}\rho_ie^{i\int_0^t G^\dagger(s)\,\d s} = \rho_f\right\}. \label{eq: tau_ij}
    \end{equation}
    Then the time of mixed state evolution under a normalised non-Hermitian Hamiltonian $G$ satisfies
    \begin{equation}
         \tau_{if}^{NH}\geq \frac{1}{2}\sqrt{\log^2
         \left(\sqrt{
         \frac{\Tr\big[\rho_f^2\big]}{\Tr[\rho_i^2]}
         }\right) + \arccos^2
         \left(
         \frac{ \Tr[\rho_i \rho_f]}{\sqrt{\Tr[\rho_i^2]\Tr\big[\rho_f^2\big]}}
         \right)}.
    \end{equation}
\end{corollary} 

Observe that the above QSL differs from that of Corollary~\ref{cor: speed limit Lindblad} only by a factor of $\frac{1}{2}$. This originates from the fact that we fix the norm of different objects - in the first case it is a superoperator, in the second it is an ordinary operator with an adjoint action on quantum states. Moreover, it is possible for the set in \eqref{eq: tau_ij} to be empty, e.g.\ when $\rho_i, \rho_f$ have different ranks. We then have that $\tau_{if}^{NH}$ is infinite as $\rho_f$ is not reachable via non-Hermitian dynamics from $\rho_i$.

Both Theorem~\ref{thm: speed limit pure} and Corollary~\ref{cor: speed limit non-Herm} may be compared with a tight speed limit for non-Hermitian dynamics of mixed states derived in \,\cite{hornedal_geometrical_2025}. 
While both results admit geometric formulations, the geometries arise from different principles. The present work employs a geometry adapted to the available dynamical resource, namely the norm of the generator. By contrast, the speed limit of \,\cite{hornedal_geometrical_2025} is formulated in terms of a geometry associated with state distinguishability through the Uhlmann fidelity, and is not tied to a specific constraint on the generator.

\section{Proof of Theorem~\ref{thm: speed limit pure} and Corollaries~\ref{cor: speed limit Lindblad}-\ref{cor: speed limit non-Herm}}

Proof of
Theorem~\ref{thm: speed limit pure}
presented in this Section 
is based on the following lemma.
\subsection{Lemma \ref{lemm: dist(id, S)}}

\begin{lemma}\label{lemm: dist(id, S)}
    Let $m\in\mathbb{N}_+$ and let $M:=GL(m,\mathbb{C})$ be equipped with a standard right-invariant\footnote{The choice of the left-invariant metric is completely equivalent for our result since taking an inverse provides an isometry between the two spaces.} Riemannian metric defined on $T_\Phi M \cong \mathfrak{gl}(m,\mathbb{C})$ as follows
    \begin{equation} 
        g_\Phi: T_\Phi M\times T_\Phi M \to \mathbb{R},\quad g_\Phi(X,Y):=\mathrm{Re}\,\Tr\!\left[(X\Phi^{-1})^\dagger(Y \Phi^{-1})\right],
    \end{equation}
    and denote the derived geodesic distance $\mathrm{dist}\!: M\times M \to \mathbb{R}_{\geq 0}$. For $v_i, v_f \in \Omega := \mathbb{C}^m\setminus\{0\}$, define the set of restricted maps
    \begin{equation}
        S_{if} := \left\lbrace \Phi \in M |\, \Phi v_i = v_f\right\rbrace.
    \end{equation}
    Define $\lambda := \log( \frac{\norm*{v_f}}{\norm{v_i}})$ and $\theta := \arccos(\Re \frac{v_i^\dagger v_f}{\norm{v_i}\norm*{v_f}})$. Then
    \begin{equation}
        \dist{\id,S_{if}} = \sqrt{\lambda^2 + \theta^2},
    \end{equation}
    and this distance is realised by a path $\Phi:[0,1]\to M$ satisfying
    \begin{equation}
        \dot{\Phi}(t) = G(t)\Phi(t),
    \end{equation}
    where
    \begin{equation}
        G(t):= \frac{\dot{v}(t) v^\dagger (t)}{\norm{v(t)}^2}, \quad v(t) :=
        \begin{cases}
            \norm{v_i}e^{\lambda t}\left( \frac{\sin((1-t)\theta)}{\sin(\theta)} \frac{v_i}{\norm{v_i}} + \frac{\sin(t\theta)}{\sin(\theta)} \frac{v_f}{\norm*{v_f}}\right), &\theta\in (0,\pi),\\
            \norm{v_i}e^{\lambda t}\left( \cos(\theta t) \frac{v_i}{\norm{v_i}} + \sin(\theta t) u_i^\perp \right), &\theta = \pi, \\
            e^{\lambda t} v_i, &\theta=0,\\
        \end{cases}
    \end{equation}
    and $u_i^\perp$ is any normalised vector satisfying $\Re u_i^\dagger u_i^\perp=0$ (the solution is non-unique).
    
    Moreover, $g_\Phi$ induces a metric on $T\Omega$:
    \begin{equation}
        g_v: T\Omega\times T\Omega\to \mathbb{R}, \quad g_v(A,B):=\frac{\Re(A^\dagger B)}{\norm{v}^2},
    \end{equation}
    whose geodesic distance $\mathrm{dist}_\Omega:\Omega\times\Omega\to \mathbb{R}_{\geq 0}$ satisfies
    \begin{equation}
        \mathrm{dist}_\Omega(v_i,v_f)=\dist{\id, S_{if}}.\label{dist_pytagor}
    \end{equation}
    
\end{lemma}
\vspace{.5cm}

Before proving the Lemma we will need to recall
some  geometric notions.

\subsection{Principal bundle}\label{sec: bundle}
Let $v_i,v_f\in \Omega:=\mathbb{C}^m\setminus\{0\}$ and define 
\begin{equation}
    \text{Stab}(v_i):=\left\lbrace \Phi\in GL(m,\mathbb{C})|\, \Phi v_i = v_i \right\rbrace
\end{equation}
The quadruple $(M,\Omega,\pi,\alpha)$ forms a principal bundle, with the projection map $\pi:M\to \Omega,\,\pi(\Phi):=\Phi v_i = v$ and the group action $\alpha: M\times \text{Stab}(v_i) \to M,\,  \alpha(\Psi, \Phi) := \Psi \circ \Phi$. We define the vertical space $V_\Phi M := \mathrm{ker}(\d \pi_\Phi)$ and its orthogonal complement, the horizontal space
\begin{equation}
    H_\Phi M := \mathrm{ker}(\d \pi_\Phi)^\perp =\{ X\in T_\Phi M | \, \exists w\in \mathbb{C}^{m} : X=w v^\dagger \Phi  \}. 
\end{equation}

We are interested in a geodesic distance $\dist{\id,S_{if}} = \dist{\id,\pi^{-1}(v_f)}$. It is well known that this distance is realised by a completely horizontal geodesic \cite{carmo_riemannian_2013}. Since the horizontal space $H_\Phi M$ is isomorphic to $T_v\Omega$, the latter inherits the Riemannian metric, which takes particularly simple form. Let $h_{v,\Phi}:T_v \Omega \to H_{\Phi} M,\, h_{v,\Phi}(A) = \left(\d \pi_{\Phi}\vert_{H_{\Phi}}\right)^{-1}(A)$, where $\Phi\in\pi^{-1}(v)$, be a horizontal lift realising this isomorphism. The inherited metric is
\begin{equation}
    g_v:T_v \Omega \times T_v \Omega \to \mathbb{R},\quad g_v(A, B) := g_{\Phi}(h_{v,\Phi}(A),h_{v,\Phi}(B))= \frac{\Re(A^\dagger B)}{\norm{v}^2},
\end{equation}
which indeed does not depend on the choice of $\Phi$, hence it is well defined.
It follows that the associated distance function $\text{dist}_\Omega\!:\Omega\times\Omega \to \mathbb{R}_{\geq 0}$ satisfies $\text{dist}_\Omega(v_i,v_f) = \dist{\id,\pi^{-1}(v_f)}$.

\subsection{Proof of Lemma~\ref{lemm: dist(id, S)}}
It is enough to compute $\text{dist}_\Omega(v_i,v_f)$. Consider the decomposition $v = r u$, where $r := \norm{v}$. Observe that the fundamental form associated with the inner product $g_v$ satisfies
\begin{equation}
    \d s^2 = \frac{\norm{\d v}^2}{\norm{v}^2} = \d( \log r)^2 + \norm{\d u}^2 = \d (\log r)^2 + \d s_{S^{(2 m -1)}}^2.
\end{equation}
Thus $\Omega \cong \mathbb{R} \times S^{(2 m -1 )}$ equipped with a product Riemannian metric, which makes geodesics particularly simple - they are formed by geodesics on each of the components. Define ${r_i:=\norm{v_i}, r_f:=\norm{v_f}, u_i:=\frac{v_i}{r_i}, u_f:=\frac{v_f}{r_f}}$, and
\begin{equation}
\lambda:=\log\frac{r_f}{r_i},
\qquad
\theta:=\arccos\!\left(\Re u_i^\dagger u_f\right).
\end{equation}
The shortest geodesic from $v_i$ to $v_f$ is given by
\begin{equation}
    v\!: [0,1]\to\Omega,\quad v(t):=r_i e^{\lambda t}u(t),
\end{equation}
where $u(t)$ is the real spherical geodesic
\begin{equation}
    u(t)=\cos(\theta t)\, u_i + \sin(\theta t)\, u_i^\perp, \qquad u_i^\perp := (u_f - \cos\theta\, u_i)/\sin\theta.
\end{equation}
provided that $\theta \not\in \{0,\pi\}$. If $\theta=0$ then $u(t)\equiv u_i$, and if $\theta=\pi$ then $u_i^\perp$ can be taken to be an arbitrary normalized vector satisfying $\Re u_i^\dagger u_i^\perp = 0$.
The length of $v(t)$ is hence equal
\begin{equation}
    \mathrm{dist}_\Omega (v_i,v_f) = \sqrt{\lambda^2 + \theta^2}.
\end{equation}
The geodesic $\Phi:[0,1]\to M$ realising $\dist{\id,S_{if}}$ satisfies then
\begin{equation}
    \dot{\Phi}(t) := h_{v(t),\Phi(t)}(\dot{v}(t)) = \frac{\dot{v}(t)v(t)^\dagger}{\norm{v(t)}^2}  \Phi(t),
\end{equation}
and the claim follows. \qed

\subsection{Proof of Theorem~\ref{thm: speed limit pure}}

    Define
    \begin{equation}
        \mathfrak{G}_{if}(\mathcal{H}):=\left\lbrace G \in \mathfrak{G}(\mathcal{H}) | \exists t\geq 0: \mathbf{T}e^{-i\int_0^t G(s)\,\d s}\ket{\psi_i} = \ket*{\psi_f} \right\rbrace.
    \end{equation}
    Given $G\in\mathfrak{G}_{if}(\mathcal{H})$ consider its integral curve passing through $\psi_i$ at time $t=0$, namely $\ket{\psi_G}: [0,\infty) \to \mathcal{H}, {\ket{\psi_G(t)} = \mathbf{T}e^{-i\int_0^t G(s)\,\d s}\ket{\psi_i}}$, and let $t_{G}\geq 0$ be the smallest time for which $\ket*{\psi_G(t_G)}=\ket*{\psi_f}$. Because $G$ is bounded, we have that $\ket{\psi_G(t)}\neq 0$ for all $t\in[0,\infty)$, and therefore $[0,\infty)\mapsto \ket{\psi_G(t)}$ is a path on $\mathcal{H}\setminus\{0\} \sim \Omega:=\mathbb{C}^{n}\setminus\{0\}$. One then arrives at the following chain of inequalities
    \begin{equation}
        t_G \geq \frac{\mathrm{dist}_\Omega(\psi_i,\psi_f)}{\frac{1}{t_G} \int_0^{t_G} \sqrt{g_{\psi_G(t)}(\dot{\psi}_G(t),\dot{\psi}_G(t))}\,\d t} = \frac{\mathrm{dist}_\Omega(\psi_i,\psi_f)}{\frac{1}{t_G} \int_0^{t_G} \frac{\norm{\dot{\psi}_G(t)}}{\norm{\psi_G(t)}}\,\d t} \geq 
        \frac{\mathrm{dist}_\Omega(\psi_i,\psi_f)}{\frac{1}{t_G} \int_0^{t_G} \norm{G(t)}_{op}\,\d t} \geq
        \frac{\mathrm{dist}_\Omega(\psi_i,\psi_f)}{\frac{1}{t_G} \int_0^{t_G} \unorm{G(t)}\,\d t} = \mathrm{dist}_\Omega(\psi_i,\psi_f),
    \end{equation}
    where the first inequality is a triangle inequality for the metric on $T\Omega$ from Lemma~\ref{lemm: dist(id, S)} in dimension $m=n$. The second one (equality) uses the definition of this metric, third comes from the definition of $\ket{\psi_G(t)}$ and the definition of operator norm, fourth from Fact~\ref{fact: operator norm}, and the fifth (equality) from the assumption $\unorm{G(t)}=1$. 
    The claim then follows by Lemma~\ref{lemm: dist(id, S)} and the fact that all unitarily invariant norms agree on rank-1 operators. \qed

\subsection{Proof of both Corollaries}

\begin{proof}[Proof of Corollary~\ref{cor: speed limit Lindblad}]
    Observe that $(\mathbb{C}^{n^2}\setminus\{0\},\norm{\bullet})$ is isomorphic to $(\mathbb{C}^{n\times n}\setminus\{0\},\norm{\bullet}_{HS})$. The proof is then completely analogous to the proof of Theorem~\ref{thm: speed limit pure} using Lemma~\ref{lemm: dist(id, S)} in dimension $m=n^2$. We however cannot claim the tightness of the bound as the generator of the saturating dynamics may not be a Lindbladian. In particular, Lindbladian dynamics is trace preserving while the identified geodesic has a constant trace iff $\lambda=\theta=0$. In order to tighten the bound it is possible to restrict the metric on $T\Omega$ to the hyperplane $\Tr[\rho]=1$, however the resulting distance function is not explicit, except for the trivial cases.
\end{proof}

\begin{proof}[Proof of Corollary~\ref{cor: speed limit non-Herm}]
    Similarly as for the Corollary~\ref{cor: speed limit Lindblad}, the proof follows from Lemma~\ref{lemm: dist(id, S)} in dimension $m=n^2$. The additional factor of $1/2$ comes from the following inequality for any $G\in \mathfrak{G}(\mathcal{H})$
    \begin{equation}
        g_{\rho(t)}(\dot{\rho}(t),\dot{\rho}(t))^{1/2} = \frac{\norm{\dot{\rho}(t)}_{HS}}{\norm{\rho(t)}_{HS}} = \frac{\norm{G(t) \rho(t) - \rho(t) G^\dagger (t) }_{HS}}{\norm{\rho(t)}_{HS}} \leq 2\norm{G(t)}_{op} \leq 2\unorm{G(t)}= 2.
    \end{equation}
\end{proof}

\begin{remark}
Observe that in all the above proofs we only used the fact that the dynamics takes place on a certain Riemannian manifold $\Omega$ on which the geodesic distance function is known explicitly. In principle, the first part of the Lemma~\ref{lemm: dist(id, S)} regarding $GL(m,\mathbb{C})$ could have been therefore omitted. However, we decided to preserve this part of the lemma for two reasons. First, as explained in Section~\ref{sec: bundle}, it provides a motivation for the particular choice of the metric on $T\Omega$, as the one arising from the most natural metric on the space of maps $GL(m,\mathbb{C})$. Second, it constitutes an interesting result on its own, since while the geodesics on $GL(m,\mathbb{C})$ are well understood\cite{andruchow_left_2014}, much less is known about the distance function there\cite{kupferman_strain_2016}.
\end{remark}

\section{Exemplary applications}
\subsection{Pure states}
Consider $\ket{\psi_i},\ket{\psi_f} \in \mathcal{H}$, $\norm{\psi_i}=\norm{\psi_f}=1$. From Theorem~\ref{thm: speed limit pure} we have that
\begin{equation}
    \tau_{if}^\text{pure} = \theta := \arccos(\Re \braket{\psi_i}{ \psi_f}),\label{eq: QSL normed}
\end{equation}
and for simplicity let us assume that $\theta\in(0,\pi)$. This time is attained by an evolution generated by a time-dependent, genuinely non-Hermitian Hamiltonian 
\begin{equation}
    G^*(t):= i \ketbra{\dot{\psi}(t)}{\psi(t)},\quad \ket{\psi(t)}:= \frac{\sin(\theta-t)}{\sin\theta}\ket{\psi_i}+\frac{\sin(t)}{\sin\theta}\ket{\psi_f},\label{eq: pure generator}
\end{equation}
which coincides with the previously derived geodesic in the space of minimal purifications of quantum states with fixed rank\cite{hornedal_geometrical_2025, spehner_bures_2025}.

If one is interested in reaching state $\ket{\psi_f}$ up to an overall phase, the minimal time to do so is given by a Fubini-Study distance between states $\ket{\psi_i}$ and $\ket{\psi_f}$ (also known as Bures angle \cite{BZ17}
when $\ket{\psi_i},\ket{\psi_f}$ are seen as purifications\cite{uhlmann_transition_1976, jozsa_fidelity_1994, hornedal_geometrical_2025})
\begin{equation}
    \theta_{FS} :=\min_{\varphi} \arccos(\Re e^{i\varphi}\braket{\psi_i}{ \psi_f}) = \arccos(\abs{\braket{\psi_i}{\psi_f}}),\label{eq: QSL non-Herm}
\end{equation}
This corresponds to the adjustment of the phase of the final state $\ket{\psi_f} \to \ket{\bar{\psi}_f}:=e^{-i\mathrm{Arg} \braket{\psi_i}{\psi_f}}\ket{\psi_f}$ in \eqref{eq: QSL normed}. The generator of the saturating dynamics \eqref{eq: pure generator} acts entirely within the subspace spanned by the initial and final states. On this subspace it takes a particularly simple form when written in the basis $\{\ket{\psi_i},\ket{\psi_i^\perp}\}$, where  $\ket{\psi_i^\perp} :=\frac{\ket{\bar{\psi}_f} - \ket{\psi_i}\braket{\psi_i}{\bar{\psi}_f}}{\sqrt{1-{\braket{\psi_i}{\bar{\psi}_f}}^2}}$, namely
\begin{equation}
    G^*(t)=\frac{1}{2}\sigma_y + \frac{i}{2} e^{-i \sigma_y t} \sigma_x e^{i \sigma_y t},
\end{equation}
where $\sigma_x, \sigma_y$ are Pauli matrices.
This time-optimal non-Hermitian dynamics may be compared with the seminal Mandelstam-Tamm QSL for Hermitian dynamics\cite{mandelstam_uncertainty_1991}
\begin{equation}
    \tau^{Herm}_{if} \geq \frac{\theta_{FS}}{\Delta E},
\end{equation}
where $\Delta E$ is the energy uncertainty, necessarily equal in both initial and final state. From Popoviciu's inequality on variances\cite{popoviciu_sur_1935} it follows that $\Delta E \leq (E_{max}-E_{min})/{2}$. Now, if the norm of the Hamiltonian generating the dynamics satisfies $\norm{H}_{p}=1$ we have $E_{max}-E_{min}\leq 2^{1-1/p}$ and
\begin{equation}
    \tau^{Herm}_{if} \geq 2^{1/p}\theta_{FS}.
\end{equation}
This bound is saturated by dynamics generated by the Hamiltonian acting entirely in the subspace spanned by the initial and final states, which takes the following form when written in the basis $\{\ket{\psi_i},\ket{\psi_i^\perp}\}$\cite{carlini_timeoptimal_2006, brody_optimum_2006}
\begin{equation}
    H = \frac{1}{2^{1/p}}\sigma_y.
\end{equation}
We hence observe that allowing for non-Hermitian generators provides a speed-up of factor $2^{1/p}$, as both bounds are tight. This speed-up arises thanks to the continuous rotation of the anti-Hermitian component of the optimal generator.
Along the optimal trajectory, the velocity induced by such a rotating anti-Hermitian component remains aligned with the velocity generated by the Hermitian component. The two then add constructively, yielding the maximal enhancement of the evolution speed. We stress that this result is valid in the Hilbert space of an arbitrary dimension $n$, but the optimal dynamics takes place in its $2$-dimensional subspace.

\subsection{States with equal purity}
Purity of a normalised quantum state is defined as $\gamma(\rho):=\Tr[\rho^2]$, while the relative purity of $\rho$ wrt $\sigma$ is given by ${f(\rho|\sigma):=\Tr[\rho\sigma]/\Tr[\sigma^2]}$. Observe that for $\rho_i, \rho_f$ of equal purity, the QSLs from Corollaries~\ref{cor: speed limit Lindblad}-\ref{cor: speed limit non-Herm} simplify to functions of their relative purity
\begin{align}
    \tau_{if}^{Lind} &\geq \arccos\left( \frac{\Tr[\rho_f \rho_i]}{\Tr[\rho_i^2]}\right), \\ 
    \tau_{if}^{NH} &\geq \frac{1}{2}\arccos\left( \frac{\Tr[\rho_f \rho_i]}{\Tr[\rho_i^2]}\right).
\end{align}
Interestingly, the first of these bounds coincides with the recent result by Srivastav, Pandey, Mohan, and Pati for the choice of operator or Hilbert-Schmidt norm (cf. (7) in \,\cite{srivastav_family_2025} and preceding discussion). However, in a general case of non-equal purities, our result provides a tighter bound, as it is enlarged by a term accounting for the change of purity. This term, in turn, can be recovered from the result by Uzdin and Kosloff\cite{uzdin_speed_2016} by upper-bounding their Eq.(9) with $2\tau_{if}^{Lind}$. Our result makes precise this apparent orthogonality of purity change and Liouville space rotation.

\subsection{Replacer dynamics and dephasing Lindbladian}
Let $n=2$ and $\ket{\psi}\in\mathcal{H}$ be a normalised state. Pick $\ket{\psi_\perp}$ such that $\braket{\psi}{\psi_\perp}=0$ and consider the following Lindblad generator
\begin{equation}
    L^R(\rho)=\frac{1}{\sqrt{2}}\sum_{i=1,2}\left( L_i\rho L_i^\dagger - \frac{1}{2}\{L_i^\dagger L_i, \rho\} \right), \label{eq: replacer}
\end{equation}
where $L_1 = \ketbra{\psi}{\psi}, L_2 = \ketbra{\psi}{\psi_\perp}$, and $\norm{L}_{op}=1$. It generates the following replacer dynamics $\rho(t) = e^{-\frac{t}{\sqrt{2}}}\rho_0+(1-e^{-\frac{t}{\sqrt{2}}})\sigma$, where $\rho_0\in \mathcal{D}(\mathcal{H})$ is an initial state and $\sigma=\ketbra{\psi}{\psi}$, which is a straight line inside the Bloch ball.
Let $\rho_0 = \ketbra{+}{+}$, and let $\ket{\psi} = \cos(\frac{\theta_B}{2})\ket{0} - \sin(\frac{\theta_B}{2})\ket{1}$ for some $\theta_B\in(0,\pi)$. Then, at time $t^R=\sqrt{2}\log(1+\frac{1}{\sin(\theta_B)})$ the trajectory passes through a classical state $\omega = \frac{1}{2}\left( \mathbbm{1} + \tan(\frac{\pi}{4}-\frac{\theta_B}{2})\sigma_z\right)$.

\begin{figure}[b]
    \begin{center}
        \input{figures/comparison.pgf}
    \end{center}
    \caption{Comparison of the arrival times $t^R, t^D$ (orange, green) at the state $\omega$ from the state $\ketbra{+}$, under the single qubit dynamics generated by \eqref{eq: replacer} and \eqref{eq: dephasing} respectively, with the derived QSL \eqref{eq: QSL replacer} (black). The divergence of $t^R$ at $\theta_B =0$ is caused by the fact that the final state $\omega$ becomes pure, while the divergence of $t^D$ at $\theta_B = \pi/2$ comes from the fact that $\omega$ becomes maximally mixed - they both are being reached only asymptotically by the corresponding dynamics. For $\theta_B \in (\pi/2,\pi)$ the plot is symmetric with respect to $\pi/2$.
    Inset: trajectories of the replacer \eqref{eq: replacer} (orange) and dephasing \eqref{eq: dephasing} (green) dynamics in a Bloch ball.}
    \label{fig: replacer}
\end{figure}

Consider now the following dephasing Lindblad generator,
\begin{equation}
    L^{D}(\rho):=\frac{1}{1+\sqrt{1+\chi^2}}\left(-i \left[\frac{\chi}{2}\sigma_y,\rho\right] + \sigma_z \rho \sigma_z - \rho \right), \label{eq: dephasing}
\end{equation}
where $\chi$ solves equation:
$\tan(\frac{\pi}{4}-\frac{\theta_B}{2})=-e^{-\mathrm{arctanh}\left({\sqrt{1-\chi^2}}\right)/\sqrt{1-\chi^2}}\sign(\chi)$, and both the square root and the arctanh are to be understood as their analytic continuations. It is normalised as $\norm{L^D}_{op}=1$ and generates the dynamics taking the initial state
$\rho_0$ to the final one $\omega$ in time $t^D=\frac{1+\sqrt{1+\chi^2}}{\sqrt{1-\chi^2}}\mathrm{arctanh}\left(\sqrt{1-\chi^2}\right)$.

The speed limit  derived here provides a lower bound for the shortest evolution time,
\begin{equation}
    \min(t^R,t^D) \geq \sqrt{ \frac{1}{4}\log^2\left(1+\sin\theta_B\right) + \arccos^2 \left( \frac{\sqrt{1+\sin\theta_B}}{2} \right) }.\label{eq: QSL replacer}
\end{equation}
We illustrate the tightness of this bound in Fig.~\ref{fig: replacer}. Note that we have chosen here to normalise the Lindblad generators according to their operator norm. Among all the unitarily invariant norms, our QSL is the tightest for such a choice, which follows from the Fact~\ref{fact: operator norm}.

\subsection{Time-optimal transport in a qubit}
Consider the following norm defined for on the cone of Lindbladians
\begin{equation}
    \norm{L}_{HK} := \frac{1}{2}\Tr[\abs{H}] + \frac{1}{2}\Tr[K], \label{eq: AD}
\end{equation}
where $H$ is the traceless Hamiltonian associated with $L$, and $K$ is the Kossakowski matrix in the basis of generalised normalised Gell-Mann matrices (see \,\cite{krakow_quantum} for details).

The amplitude-damping Lindbladian normalised to $\norm{L^\text{AD}}_{HK} = 1$ is defined as
\begin{equation}
    L^\text{AD}(\rho) = 2\left( \ketbra{1}{0}\!\rho\!\ketbra{0}{1} - \frac{1}{2}\lbrace \ketbra{0},\rho \rbrace\right).
\end{equation}
It encodes the time-optimal classical Markovian dynamics between the states $\eta_{1}=\ketbra{0}$ and ${\eta_{-1}=\ketbra{1}}$, under the condition that the corresponding purely dissipative Lindbladian is normalised as above --
see \,\cite{krakow_classical}. On the other hand, the time-optimal purely quantum (unitary) dynamics between $\eta_1$ and $\eta_{-1}$, under the normalisation condition, is generated by
\begin{equation}
    L^\text{H}(\rho) = -i[\sigma_y,\rho].
\end{equation}

\begin{figure}[!ht]
    \begin{center}
        \input{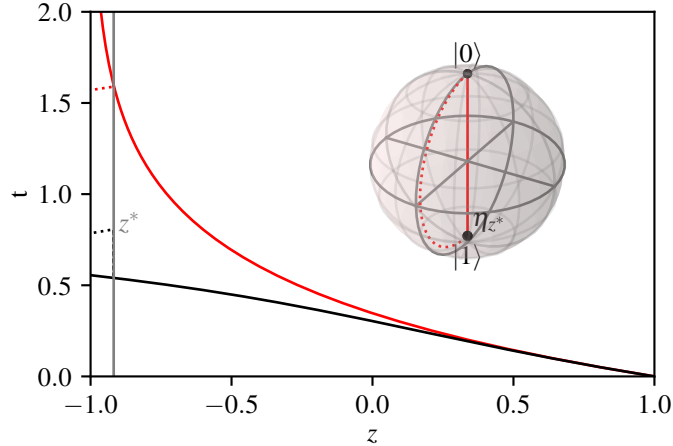}
    \end{center}
    \caption{Comparison of the arrival times $t^{AD}, t^z$ (red solid, red dotted) at the state $\eta_z$ from the state $\eta_1=\ketbra{0}$, under the single qubit dynamics generated by $L^\text{AD}$ and $L^{\alpha_z}$ respectively, with the QSL \eqref{eq: QSL opt} correspondingly adjusted for their norms  (black solid, black dotted). At $z = z^*$ (grey line) there is a transition between $\alpha \neq 0$ and $\alpha =0$ ($L^\text{AD}$) being optimal. Data courtesy to G. M\"uller-Rigat (Jagiellonian University).
    Inset: the equal time trajectories generated by $L^\text{AD}$ (red solid) and $L^{\alpha\neq0}$ (red dotted) from $\eta_1$ to the critical $\eta_{z^*}$.}
    \label{fig: time-opt}
\end{figure}

It is natural to consider the mixture $L^\alpha = \alpha L^\text{H} + (1-\alpha) L^\text{AD}$ as generator of the time-optimal evolution between the states $\eta_1$ and $\eta_z$ with $z\in(-1,1)$. 
Numerical simulations suggest \cite{krakow_quantum} 
that for $z \ge z_* \approx -0.92$ the optimal choice is $\alpha=0$, i.e.\ the purely dissipative evolution.
In such a case the time it requires to reach $\eta_z$ equals
\begin{equation}
    t^\text{AD} = 
    -\frac{1}{2}\log( \frac{1+z}{2}).
\end{equation}
For $z\le z_*$ the Hamiltonian contribution becomes non-trivial as the optimal value of the weight $\alpha$ is positive. In Fig.~\ref{fig: time-opt} we compare the time $t^{z}$, defined as the shortest time in which it is possible to reach $\eta_z$ from $\eta_1$ using $L^{\alpha}$ for some $\alpha$, with  bound \eqref{eq: limit Lindblad} adjusted for the operator norm of the optimal generator
\begin{equation}
    t^z \geq \frac{\sqrt{2}}{\sqrt{17\alpha_z^2-18\alpha_z+9+(1-\alpha_z)\sqrt{65\alpha_z^2-98\alpha_z+49}}}\sqrt{\log^2\left( \sqrt{\frac{1+z^2}{2}}\right) + \arccos^2\left(\frac{1+z}{\sqrt{2(1+z^2)}}\right)},\label{eq: QSL opt}
\end{equation}
where we made explicit that
the optimal weight 
$\alpha$ depends on $z$.
From the figure we observe that the provided speed limit is asymptotically tight for the vanishing displacement $1-z$.

\subsection{Typical dynamics}
To visualise proposed QSLs we performed numerical simulations of random non-Hermitian Hamiltonians and random purely dissipative Lindblad generators -- see Fig.~\ref{fig: simulations}.
We observe very good performance of the QSL for non-Hermitian evolution of pure states (Theorem~\ref{thm: speed limit pure}), while the bound for Lindblad-generated dynamics (Corollary~\ref{cor: speed limit Lindblad}) deteriorates with the size of the Hilbert space $n$. This can be hinted by the fact that the dynamics saturating the bounds is generated by a rank~1 generator, and the minimal rank of a non-zero, purely dissipative Lindblad generator equals at least $2(n-1)$ (cf Theorem~4 in  \,\cite{amato_number_2024} and the rank-nullity theorem).

\begin{figure}[!ht]
    \begin{subfloat}[]
    {
        \input{figures/optimized_parallel_n2.pgf}
    }
    \end{subfloat}
    \begin{subfloat}[]
    {
        \input{figures/optimized_parallel_n4.pgf}
    }
    \end{subfloat}\hfill
    \begin{subfloat}[]
    {
        \input{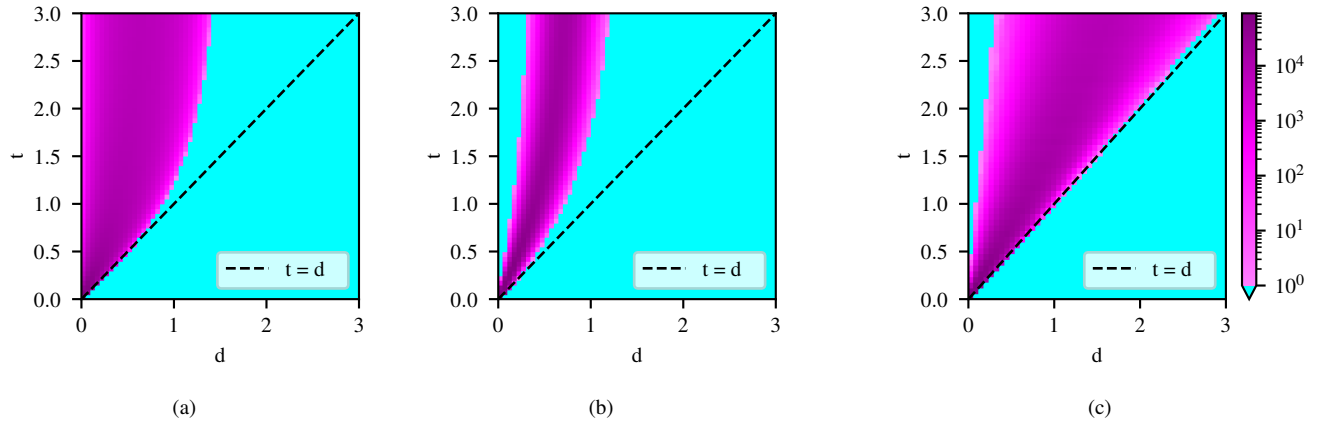}
    }
    \end{subfloat}
    \caption{Tightness of bounds derived: 
    Value of bin $(d,t)$ corresponds to the probability 
    that a random initial state $v_i$ evolved by a random generator 
    for the first time crossed at time $t$ the sphere centred 
    at $v_i$ of radius $d$ measured with respect to the distance (\ref{dist_pytagor}).
    Purely dissipative Lindblad dynamics in Hilbert spaces of dimension (a) $n=2$ and (b) $n=4$. Each plot corresponds to choices of $300$ random initial conditions drawn from the Wishart ensemble and normalised to unit trace, and $300$ random Kossakowski matrices drawn from the Wishart ensemble and normalised in such a way that $\norm{L}_{op}=1$. The distance $d$ corresponds to the bound \eqref{eq: limit Lindblad}. Observe that the QSL becomes looser in higher dimensions. 
    Non-Hermitian evolution in Hilbert space of dimension (c)  $n=4$. 
    The plot corresponds to choices of $300$ random initial pure states taken 
    to be $n$-dimensional complex standard random vectors and 
    $300$ random Hamiltonians drawn from the Ginibre ensemble and normalised to $\norm{G}_{op}=1$. 
    The distance $d$ corresponds to the minimal time \eqref{eq: limit pure}.
    In this case, the QSL remains accurate also in higher dimensions.
    }
    \label{fig: simulations}
\end{figure}

\section{Concluding Remarks}

In this work we ask how quickly one quantum state (pure or mixed) can be converted into another when the generator of the evolution is constrained only by its size, measured by a unitarily invariant norm. This approach covers Lindblad generators and non-Hermitian Hamiltonians in a Hilbert space of arbitrary dimension $n$.
Using a geodesic distance on the space of states constructed from the natural geometry of the space of automorphisms, we derive a general bound for the evolution time that does not depend on which unitarily invariant norm is chosen.
Exact results are derived for the time evolution of pure states under a normalised non-Hermitian Hamiltonian. The optimal generator is time-dependent and acts only in the two-dimensional subspace spanned by the initial and final states. For Hamiltonians with fixed Schatten-$p$ norm, the time is reduced by a factor $2^{1/p}$ compared with the Mandelstam–Tamm bound, so non-Hermitian generators provide a genuine speed-up. Numerical simulations for random generators illustrate that the pure-state bound stays accurate in higher dimensions, while the Lindblad bound becomes looser as $n$ grows.
Although other bounds for Lindblad dynamics were obtained in a parallel work \cite{krakow_quantum}, the problem of deriving exact results in this case and finding the saturating generator remains open.

The approach presented here can also be discussed from the point of view of quantum resource theory \cite{CG19}.
The norm of the Hamiltonian operator $H$ or of the Lindblad generator $L$ has the dimension of inverse time and can be interpreted as a valuable resource that one has to pay for. Thus the key question is how to reach the final state 
$\rho_f$ from a certain initial state $\rho_i$
in the shortest time $T$, given the fixed resources $\unorm{L}=1$. Alternatively, the travel time can be fixed, $T=1$, and one can ask for the optimal path that minimises the resources, measured by $\unorm{L}$.
In the cases where the generator saturating the derived bound is known, it solves both dual problems.

\section*{Acknowledgements}
The authors would like to thank Niklas H\"ornedal for his help in finding the right approach to the problem  and simplifying the final proof. 
This work would not have been possible without his assistance.
It is a pleasure to thank Alfonso  de Bobadilla,
Dariusz Chru{\'s}ci{\'n}ski,  Mykhailo Hontarenko,
  Guillem M\"uller-Rigat  and Fereshte Shahbegi for numerous  discussions and fruitful correspondence.

OAP thanks Aur\'elia Chenu and Adolfo del Campo for their support of the work leading to this manuscript. OAP also thanks Kyunghoon Han for the early-stage discussions.

This research was funded in whole, or in part, by the Luxembourg National Research
Fund (FNR), grant reference C24/MS/18969168/MOS. For the purpose of open access, and
in fulfilment of the obligations arising from the grant agreement, the author has applied
a Creative Commons Attribution 4.0 International (CC BY 4.0) license to any Author
Accepted Manuscript version arising from this submission. 

K{\.Z} acknowledges financial support from the European Union through ERC Advanced Grant TAtypic,  Project No. 101142236.

\bibliography{bibliography}

\end{document}

%% file: figures/optimized_parallel_n2.pgf
\begingroup%
\makeatletter%
\begin{pgfpicture}%
\pgfpathrectangle{\pgfpointorigin}{\pgfqpoint{2.077478in}{2.077478in}}%
\pgfusepath{use as bounding box, clip}%
\begin{pgfscope}%
\pgfsetbuttcap%
\pgfsetmiterjoin%
\definecolor{currentfill}{rgb}{1.000000,1.000000,1.000000}%
\pgfsetfillcolor{currentfill}%
\pgfsetlinewidth{0.000000pt}%
\definecolor{currentstroke}{rgb}{1.000000,1.000000,1.000000}%
\pgfsetstrokecolor{currentstroke}%
\pgfsetdash{}{0pt}%
\pgfpathmoveto{\pgfqpoint{0.000000in}{0.000000in}}%
\pgfpathlineto{\pgfqpoint{2.077478in}{0.000000in}}%
\pgfpathlineto{\pgfqpoint{2.077478in}{2.077478in}}%
\pgfpathlineto{\pgfqpoint{0.000000in}{2.077478in}}%
\pgfpathlineto{\pgfqpoint{0.000000in}{0.000000in}}%
\pgfpathclose%
\pgfusepath{fill}%
\end{pgfscope}%
\begin{pgfscope}%
\pgfsetbuttcap%
\pgfsetmiterjoin%
\definecolor{currentfill}{rgb}{1.000000,1.000000,1.000000}%
\pgfsetfillcolor{currentfill}%
\pgfsetlinewidth{0.000000pt}%
\definecolor{currentstroke}{rgb}{0.000000,0.000000,0.000000}%
\pgfsetstrokecolor{currentstroke}%
\pgfsetstrokeopacity{0.000000}%
\pgfsetdash{}{0pt}%
\pgfpathmoveto{\pgfqpoint{0.502394in}{0.450309in}}%
\pgfpathlineto{\pgfqpoint{1.947964in}{0.450309in}}%
\pgfpathlineto{\pgfqpoint{1.947964in}{1.938898in}}%
\pgfpathlineto{\pgfqpoint{0.502394in}{1.938898in}}%
\pgfpathlineto{\pgfqpoint{0.502394in}{0.450309in}}%
\pgfpathclose%
\pgfusepath{fill}%
\end{pgfscope}%
\begin{pgfscope}%
\pgfpathrectangle{\pgfqpoint{0.502394in}{0.450309in}}{\pgfqpoint{1.445569in}{1.488589in}}%
\pgfusepath{clip}%
\pgfsys@transformshift{0.502394in}{0.450309in}%
\pgftext[left,bottom]{\includegraphics[interpolate=true,width=1.446667in,height=1.490000in]{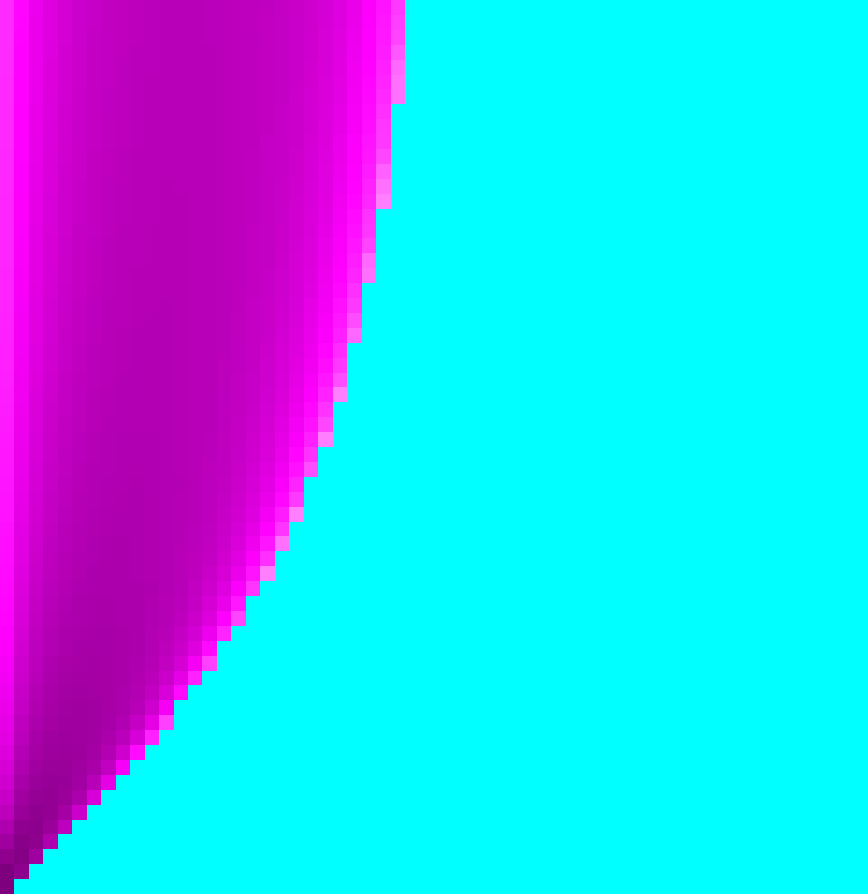}}%
\end{pgfscope}%
\begin{pgfscope}%
\pgfsetbuttcap%
\pgfsetroundjoin%
\definecolor{currentfill}{rgb}{0.000000,0.000000,0.000000}%
\pgfsetfillcolor{currentfill}%
\pgfsetlinewidth{0.803000pt}%
\definecolor{currentstroke}{rgb}{0.000000,0.000000,0.000000}%
\pgfsetstrokecolor{currentstroke}%
\pgfsetdash{}{0pt}%
\pgfsys@defobject{currentmarker}{\pgfqpoint{0.000000in}{-0.048611in}}{\pgfqpoint{0.000000in}{0.000000in}}{%
\pgfpathmoveto{\pgfqpoint{0.000000in}{0.000000in}}%
\pgfpathlineto{\pgfqpoint{0.000000in}{-0.048611in}}%
\pgfusepath{stroke,fill}%
}%
\begin{pgfscope}%
\pgfsys@transformshift{0.502394in}{0.450309in}%
\pgfsys@useobject{currentmarker}{}%
\end{pgfscope}%
\end{pgfscope}%
\begin{pgfscope}%
\definecolor{textcolor}{rgb}{0.000000,0.000000,0.000000}%
\pgfsetstrokecolor{textcolor}%
\pgfsetfillcolor{textcolor}%
\pgftext[x=0.502394in,y=0.353087in,,top]{\color{textcolor}{\rmfamily\fontsize{8.000000}{9.600000}\selectfont\catcode`\^=\active\def^{\ifmmode\sp\else\^{}\fi}\catcode`\%=\active\def
\end{pgfscope}%
\begin{pgfscope}%
\pgfsetbuttcap%
\pgfsetroundjoin%
\definecolor{currentfill}{rgb}{0.000000,0.000000,0.000000}%
\pgfsetfillcolor{currentfill}%
\pgfsetlinewidth{0.803000pt}%
\definecolor{currentstroke}{rgb}{0.000000,0.000000,0.000000}%
\pgfsetstrokecolor{currentstroke}%
\pgfsetdash{}{0pt}%
\pgfsys@defobject{currentmarker}{\pgfqpoint{0.000000in}{-0.048611in}}{\pgfqpoint{0.000000in}{0.000000in}}{%
\pgfpathmoveto{\pgfqpoint{0.000000in}{0.000000in}}%
\pgfpathlineto{\pgfqpoint{0.000000in}{-0.048611in}}%
\pgfusepath{stroke,fill}%
}%
\begin{pgfscope}%
\pgfsys@transformshift{0.984251in}{0.450309in}%
\pgfsys@useobject{currentmarker}{}%
\end{pgfscope}%
\end{pgfscope}%
\begin{pgfscope}%
\definecolor{textcolor}{rgb}{0.000000,0.000000,0.000000}%
\pgfsetstrokecolor{textcolor}%
\pgfsetfillcolor{textcolor}%
\pgftext[x=0.984251in,y=0.353087in,,top]{\color{textcolor}{\rmfamily\fontsize{8.000000}{9.600000}\selectfont\catcode`\^=\active\def^{\ifmmode\sp\else\^{}\fi}\catcode`\%=\active\def
\end{pgfscope}%
\begin{pgfscope}%
\pgfsetbuttcap%
\pgfsetroundjoin%
\definecolor{currentfill}{rgb}{0.000000,0.000000,0.000000}%
\pgfsetfillcolor{currentfill}%
\pgfsetlinewidth{0.803000pt}%
\definecolor{currentstroke}{rgb}{0.000000,0.000000,0.000000}%
\pgfsetstrokecolor{currentstroke}%
\pgfsetdash{}{0pt}%
\pgfsys@defobject{currentmarker}{\pgfqpoint{0.000000in}{-0.048611in}}{\pgfqpoint{0.000000in}{0.000000in}}{%
\pgfpathmoveto{\pgfqpoint{0.000000in}{0.000000in}}%
\pgfpathlineto{\pgfqpoint{0.000000in}{-0.048611in}}%
\pgfusepath{stroke,fill}%
}%
\begin{pgfscope}%
\pgfsys@transformshift{1.466107in}{0.450309in}%
\pgfsys@useobject{currentmarker}{}%
\end{pgfscope}%
\end{pgfscope}%
\begin{pgfscope}%
\definecolor{textcolor}{rgb}{0.000000,0.000000,0.000000}%
\pgfsetstrokecolor{textcolor}%
\pgfsetfillcolor{textcolor}%
\pgftext[x=1.466107in,y=0.353087in,,top]{\color{textcolor}{\rmfamily\fontsize{8.000000}{9.600000}\selectfont\catcode`\^=\active\def^{\ifmmode\sp\else\^{}\fi}\catcode`\%=\active\def
\end{pgfscope}%
\begin{pgfscope}%
\pgfsetbuttcap%
\pgfsetroundjoin%
\definecolor{currentfill}{rgb}{0.000000,0.000000,0.000000}%
\pgfsetfillcolor{currentfill}%
\pgfsetlinewidth{0.803000pt}%
\definecolor{currentstroke}{rgb}{0.000000,0.000000,0.000000}%
\pgfsetstrokecolor{currentstroke}%
\pgfsetdash{}{0pt}%
\pgfsys@defobject{currentmarker}{\pgfqpoint{0.000000in}{-0.048611in}}{\pgfqpoint{0.000000in}{0.000000in}}{%
\pgfpathmoveto{\pgfqpoint{0.000000in}{0.000000in}}%
\pgfpathlineto{\pgfqpoint{0.000000in}{-0.048611in}}%
\pgfusepath{stroke,fill}%
}%
\begin{pgfscope}%
\pgfsys@transformshift{1.947964in}{0.450309in}%
\pgfsys@useobject{currentmarker}{}%
\end{pgfscope}%
\end{pgfscope}%
\begin{pgfscope}%
\definecolor{textcolor}{rgb}{0.000000,0.000000,0.000000}%
\pgfsetstrokecolor{textcolor}%
\pgfsetfillcolor{textcolor}%
\pgftext[x=1.947964in,y=0.353087in,,top]{\color{textcolor}{\rmfamily\fontsize{8.000000}{9.600000}\selectfont\catcode`\^=\active\def^{\ifmmode\sp\else\^{}\fi}\catcode`\%=\active\def
\end{pgfscope}%
\begin{pgfscope}%
\definecolor{textcolor}{rgb}{0.000000,0.000000,0.000000}%
\pgfsetstrokecolor{textcolor}%
\pgfsetfillcolor{textcolor}%
\pgftext[x=1.225179in,y=0.198766in,,top]{\color{textcolor}{\rmfamily\fontsize{8.000000}{9.600000}\selectfont\catcode`\^=\active\def^{\ifmmode\sp\else\^{}\fi}\catcode`\%=\active\def
\end{pgfscope}%
\begin{pgfscope}%
\pgfsetbuttcap%
\pgfsetroundjoin%
\definecolor{currentfill}{rgb}{0.000000,0.000000,0.000000}%
\pgfsetfillcolor{currentfill}%
\pgfsetlinewidth{0.803000pt}%
\definecolor{currentstroke}{rgb}{0.000000,0.000000,0.000000}%
\pgfsetstrokecolor{currentstroke}%
\pgfsetdash{}{0pt}%
\pgfsys@defobject{currentmarker}{\pgfqpoint{-0.048611in}{0.000000in}}{\pgfqpoint{-0.000000in}{0.000000in}}{%
\pgfpathmoveto{\pgfqpoint{-0.000000in}{0.000000in}}%
\pgfpathlineto{\pgfqpoint{-0.048611in}{0.000000in}}%
\pgfusepath{stroke,fill}%
}%
\begin{pgfscope}%
\pgfsys@transformshift{0.502394in}{0.450309in}%
\pgfsys@useobject{currentmarker}{}%
\end{pgfscope}%
\end{pgfscope}%
\begin{pgfscope}%
\definecolor{textcolor}{rgb}{0.000000,0.000000,0.000000}%
\pgfsetstrokecolor{textcolor}%
\pgfsetfillcolor{textcolor}%
\pgftext[x=0.254321in, y=0.411729in, left, base]{\color{textcolor}{\rmfamily\fontsize{8.000000}{9.600000}\selectfont\catcode`\^=\active\def^{\ifmmode\sp\else\^{}\fi}\catcode`\%=\active\def
\end{pgfscope}%
\begin{pgfscope}%
\pgfsetbuttcap%
\pgfsetroundjoin%
\definecolor{currentfill}{rgb}{0.000000,0.000000,0.000000}%
\pgfsetfillcolor{currentfill}%
\pgfsetlinewidth{0.803000pt}%
\definecolor{currentstroke}{rgb}{0.000000,0.000000,0.000000}%
\pgfsetstrokecolor{currentstroke}%
\pgfsetdash{}{0pt}%
\pgfsys@defobject{currentmarker}{\pgfqpoint{-0.048611in}{0.000000in}}{\pgfqpoint{-0.000000in}{0.000000in}}{%
\pgfpathmoveto{\pgfqpoint{-0.000000in}{0.000000in}}%
\pgfpathlineto{\pgfqpoint{-0.048611in}{0.000000in}}%
\pgfusepath{stroke,fill}%
}%
\begin{pgfscope}%
\pgfsys@transformshift{0.502394in}{0.698407in}%
\pgfsys@useobject{currentmarker}{}%
\end{pgfscope}%
\end{pgfscope}%
\begin{pgfscope}%
\definecolor{textcolor}{rgb}{0.000000,0.000000,0.000000}%
\pgfsetstrokecolor{textcolor}%
\pgfsetfillcolor{textcolor}%
\pgftext[x=0.254321in, y=0.659827in, left, base]{\color{textcolor}{\rmfamily\fontsize{8.000000}{9.600000}\selectfont\catcode`\^=\active\def^{\ifmmode\sp\else\^{}\fi}\catcode`\%=\active\def
\end{pgfscope}%
\begin{pgfscope}%
\pgfsetbuttcap%
\pgfsetroundjoin%
\definecolor{currentfill}{rgb}{0.000000,0.000000,0.000000}%
\pgfsetfillcolor{currentfill}%
\pgfsetlinewidth{0.803000pt}%
\definecolor{currentstroke}{rgb}{0.000000,0.000000,0.000000}%
\pgfsetstrokecolor{currentstroke}%
\pgfsetdash{}{0pt}%
\pgfsys@defobject{currentmarker}{\pgfqpoint{-0.048611in}{0.000000in}}{\pgfqpoint{-0.000000in}{0.000000in}}{%
\pgfpathmoveto{\pgfqpoint{-0.000000in}{0.000000in}}%
\pgfpathlineto{\pgfqpoint{-0.048611in}{0.000000in}}%
\pgfusepath{stroke,fill}%
}%
\begin{pgfscope}%
\pgfsys@transformshift{0.502394in}{0.946505in}%
\pgfsys@useobject{currentmarker}{}%
\end{pgfscope}%
\end{pgfscope}%
\begin{pgfscope}%
\definecolor{textcolor}{rgb}{0.000000,0.000000,0.000000}%
\pgfsetstrokecolor{textcolor}%
\pgfsetfillcolor{textcolor}%
\pgftext[x=0.254321in, y=0.907925in, left, base]{\color{textcolor}{\rmfamily\fontsize{8.000000}{9.600000}\selectfont\catcode`\^=\active\def^{\ifmmode\sp\else\^{}\fi}\catcode`\%=\active\def
\end{pgfscope}%
\begin{pgfscope}%
\pgfsetbuttcap%
\pgfsetroundjoin%
\definecolor{currentfill}{rgb}{0.000000,0.000000,0.000000}%
\pgfsetfillcolor{currentfill}%
\pgfsetlinewidth{0.803000pt}%
\definecolor{currentstroke}{rgb}{0.000000,0.000000,0.000000}%
\pgfsetstrokecolor{currentstroke}%
\pgfsetdash{}{0pt}%
\pgfsys@defobject{currentmarker}{\pgfqpoint{-0.048611in}{0.000000in}}{\pgfqpoint{-0.000000in}{0.000000in}}{%
\pgfpathmoveto{\pgfqpoint{-0.000000in}{0.000000in}}%
\pgfpathlineto{\pgfqpoint{-0.048611in}{0.000000in}}%
\pgfusepath{stroke,fill}%
}%
\begin{pgfscope}%
\pgfsys@transformshift{0.502394in}{1.194603in}%
\pgfsys@useobject{currentmarker}{}%
\end{pgfscope}%
\end{pgfscope}%
\begin{pgfscope}%
\definecolor{textcolor}{rgb}{0.000000,0.000000,0.000000}%
\pgfsetstrokecolor{textcolor}%
\pgfsetfillcolor{textcolor}%
\pgftext[x=0.254321in, y=1.156023in, left, base]{\color{textcolor}{\rmfamily\fontsize{8.000000}{9.600000}\selectfont\catcode`\^=\active\def^{\ifmmode\sp\else\^{}\fi}\catcode`\%=\active\def
\end{pgfscope}%
\begin{pgfscope}%
\pgfsetbuttcap%
\pgfsetroundjoin%
\definecolor{currentfill}{rgb}{0.000000,0.000000,0.000000}%
\pgfsetfillcolor{currentfill}%
\pgfsetlinewidth{0.803000pt}%
\definecolor{currentstroke}{rgb}{0.000000,0.000000,0.000000}%
\pgfsetstrokecolor{currentstroke}%
\pgfsetdash{}{0pt}%
\pgfsys@defobject{currentmarker}{\pgfqpoint{-0.048611in}{0.000000in}}{\pgfqpoint{-0.000000in}{0.000000in}}{%
\pgfpathmoveto{\pgfqpoint{-0.000000in}{0.000000in}}%
\pgfpathlineto{\pgfqpoint{-0.048611in}{0.000000in}}%
\pgfusepath{stroke,fill}%
}%
\begin{pgfscope}%
\pgfsys@transformshift{0.502394in}{1.442701in}%
\pgfsys@useobject{currentmarker}{}%
\end{pgfscope}%
\end{pgfscope}%
\begin{pgfscope}%
\definecolor{textcolor}{rgb}{0.000000,0.000000,0.000000}%
\pgfsetstrokecolor{textcolor}%
\pgfsetfillcolor{textcolor}%
\pgftext[x=0.254321in, y=1.404121in, left, base]{\color{textcolor}{\rmfamily\fontsize{8.000000}{9.600000}\selectfont\catcode`\^=\active\def^{\ifmmode\sp\else\^{}\fi}\catcode`\%=\active\def
\end{pgfscope}%
\begin{pgfscope}%
\pgfsetbuttcap%
\pgfsetroundjoin%
\definecolor{currentfill}{rgb}{0.000000,0.000000,0.000000}%
\pgfsetfillcolor{currentfill}%
\pgfsetlinewidth{0.803000pt}%
\definecolor{currentstroke}{rgb}{0.000000,0.000000,0.000000}%
\pgfsetstrokecolor{currentstroke}%
\pgfsetdash{}{0pt}%
\pgfsys@defobject{currentmarker}{\pgfqpoint{-0.048611in}{0.000000in}}{\pgfqpoint{-0.000000in}{0.000000in}}{%
\pgfpathmoveto{\pgfqpoint{-0.000000in}{0.000000in}}%
\pgfpathlineto{\pgfqpoint{-0.048611in}{0.000000in}}%
\pgfusepath{stroke,fill}%
}%
\begin{pgfscope}%
\pgfsys@transformshift{0.502394in}{1.690800in}%
\pgfsys@useobject{currentmarker}{}%
\end{pgfscope}%
\end{pgfscope}%
\begin{pgfscope}%
\definecolor{textcolor}{rgb}{0.000000,0.000000,0.000000}%
\pgfsetstrokecolor{textcolor}%
\pgfsetfillcolor{textcolor}%
\pgftext[x=0.254321in, y=1.652219in, left, base]{\color{textcolor}{\rmfamily\fontsize{8.000000}{9.600000}\selectfont\catcode`\^=\active\def^{\ifmmode\sp\else\^{}\fi}\catcode`\%=\active\def
\end{pgfscope}%
\begin{pgfscope}%
\pgfsetbuttcap%
\pgfsetroundjoin%
\definecolor{currentfill}{rgb}{0.000000,0.000000,0.000000}%
\pgfsetfillcolor{currentfill}%
\pgfsetlinewidth{0.803000pt}%
\definecolor{currentstroke}{rgb}{0.000000,0.000000,0.000000}%
\pgfsetstrokecolor{currentstroke}%
\pgfsetdash{}{0pt}%
\pgfsys@defobject{currentmarker}{\pgfqpoint{-0.048611in}{0.000000in}}{\pgfqpoint{-0.000000in}{0.000000in}}{%
\pgfpathmoveto{\pgfqpoint{-0.000000in}{0.000000in}}%
\pgfpathlineto{\pgfqpoint{-0.048611in}{0.000000in}}%
\pgfusepath{stroke,fill}%
}%
\begin{pgfscope}%
\pgfsys@transformshift{0.502394in}{1.938898in}%
\pgfsys@useobject{currentmarker}{}%
\end{pgfscope}%
\end{pgfscope}%
\begin{pgfscope}%
\definecolor{textcolor}{rgb}{0.000000,0.000000,0.000000}%
\pgfsetstrokecolor{textcolor}%
\pgfsetfillcolor{textcolor}%
\pgftext[x=0.254321in, y=1.900317in, left, base]{\color{textcolor}{\rmfamily\fontsize{8.000000}{9.600000}\selectfont\catcode`\^=\active\def^{\ifmmode\sp\else\^{}\fi}\catcode`\%=\active\def
\end{pgfscope}%
\begin{pgfscope}%
\definecolor{textcolor}{rgb}{0.000000,0.000000,0.000000}%
\pgfsetstrokecolor{textcolor}%
\pgfsetfillcolor{textcolor}%
\pgftext[x=0.198766in,y=1.194603in,,bottom,rotate=90.000000]{\color{textcolor}{\rmfamily\fontsize{8.000000}{9.600000}\selectfont\catcode`\^=\active\def^{\ifmmode\sp\else\^{}\fi}\catcode`\%=\active\def
\end{pgfscope}%
\begin{pgfscope}%
\pgfpathrectangle{\pgfqpoint{0.502394in}{0.450309in}}{\pgfqpoint{1.445569in}{1.488589in}}%
\pgfusepath{clip}%
\pgfsetbuttcap%
\pgfsetroundjoin%
\pgfsetlinewidth{1.003750pt}%
\definecolor{currentstroke}{rgb}{0.000000,0.000000,0.000000}%
\pgfsetstrokecolor{currentstroke}%
\pgfsetdash{{3.700000pt}{1.600000pt}}{0.000000pt}%
\pgfpathmoveto{\pgfqpoint{0.502394in}{0.450309in}}%
\pgfpathlineto{\pgfqpoint{1.947964in}{1.938898in}}%
\pgfusepath{stroke}%
\end{pgfscope}%
\begin{pgfscope}%
\pgfsetrectcap%
\pgfsetmiterjoin%
\pgfsetlinewidth{0.803000pt}%
\definecolor{currentstroke}{rgb}{0.000000,0.000000,0.000000}%
\pgfsetstrokecolor{currentstroke}%
\pgfsetdash{}{0pt}%
\pgfpathmoveto{\pgfqpoint{0.502394in}{0.450309in}}%
\pgfpathlineto{\pgfqpoint{0.502394in}{1.938898in}}%
\pgfusepath{stroke}%
\end{pgfscope}%
\begin{pgfscope}%
\pgfsetrectcap%
\pgfsetmiterjoin%
\pgfsetlinewidth{0.803000pt}%
\definecolor{currentstroke}{rgb}{0.000000,0.000000,0.000000}%
\pgfsetstrokecolor{currentstroke}%
\pgfsetdash{}{0pt}%
\pgfpathmoveto{\pgfqpoint{1.947964in}{0.450309in}}%
\pgfpathlineto{\pgfqpoint{1.947964in}{1.938898in}}%
\pgfusepath{stroke}%
\end{pgfscope}%
\begin{pgfscope}%
\pgfsetrectcap%
\pgfsetmiterjoin%
\pgfsetlinewidth{0.803000pt}%
\definecolor{currentstroke}{rgb}{0.000000,0.000000,0.000000}%
\pgfsetstrokecolor{currentstroke}%
\pgfsetdash{}{0pt}%
\pgfpathmoveto{\pgfqpoint{0.502394in}{0.450309in}}%
\pgfpathlineto{\pgfqpoint{1.947964in}{0.450309in}}%
\pgfusepath{stroke}%
\end{pgfscope}%
\begin{pgfscope}%
\pgfsetrectcap%
\pgfsetmiterjoin%
\pgfsetlinewidth{0.803000pt}%
\definecolor{currentstroke}{rgb}{0.000000,0.000000,0.000000}%
\pgfsetstrokecolor{currentstroke}%
\pgfsetdash{}{0pt}%
\pgfpathmoveto{\pgfqpoint{0.502394in}{1.938898in}}%
\pgfpathlineto{\pgfqpoint{1.947964in}{1.938898in}}%
\pgfusepath{stroke}%
\end{pgfscope}%
\begin{pgfscope}%
\pgfsetbuttcap%
\pgfsetmiterjoin%
\definecolor{currentfill}{rgb}{1.000000,1.000000,1.000000}%
\pgfsetfillcolor{currentfill}%
\pgfsetfillopacity{0.800000}%
\pgfsetlinewidth{1.003750pt}%
\definecolor{currentstroke}{rgb}{0.800000,0.800000,0.800000}%
\pgfsetstrokecolor{currentstroke}%
\pgfsetstrokeopacity{0.800000}%
\pgfsetdash{}{0pt}%
\pgfpathmoveto{\pgfqpoint{1.232605in}{0.505864in}}%
\pgfpathlineto{\pgfqpoint{1.870186in}{0.505864in}}%
\pgfpathquadraticcurveto{\pgfqpoint{1.892408in}{0.505864in}}{\pgfqpoint{1.892408in}{0.528087in}}%
\pgfpathlineto{\pgfqpoint{1.892408in}{0.671914in}}%
\pgfpathquadraticcurveto{\pgfqpoint{1.892408in}{0.694136in}}{\pgfqpoint{1.870186in}{0.694136in}}%
\pgfpathlineto{\pgfqpoint{1.232605in}{0.694136in}}%
\pgfpathquadraticcurveto{\pgfqpoint{1.210382in}{0.694136in}}{\pgfqpoint{1.210382in}{0.671914in}}%
\pgfpathlineto{\pgfqpoint{1.210382in}{0.528087in}}%
\pgfpathquadraticcurveto{\pgfqpoint{1.210382in}{0.505864in}}{\pgfqpoint{1.232605in}{0.505864in}}%
\pgfpathlineto{\pgfqpoint{1.232605in}{0.505864in}}%
\pgfpathclose%
\pgfusepath{stroke,fill}%
\end{pgfscope}%
\begin{pgfscope}%
\pgfsetbuttcap%
\pgfsetroundjoin%
\pgfsetlinewidth{1.003750pt}%
\definecolor{currentstroke}{rgb}{0.000000,0.000000,0.000000}%
\pgfsetstrokecolor{currentstroke}%
\pgfsetdash{{3.700000pt}{1.600000pt}}{0.000000pt}%
\pgfpathmoveto{\pgfqpoint{1.254827in}{0.610803in}}%
\pgfpathlineto{\pgfqpoint{1.365938in}{0.610803in}}%
\pgfpathlineto{\pgfqpoint{1.477049in}{0.610803in}}%
\pgfusepath{stroke}%
\end{pgfscope}%
\begin{pgfscope}%
\definecolor{textcolor}{rgb}{0.000000,0.000000,0.000000}%
\pgfsetstrokecolor{textcolor}%
\pgfsetfillcolor{textcolor}%
\pgftext[x=1.565938in,y=0.571914in,left,base]{\color{textcolor}{\rmfamily\fontsize{8.000000}{9.600000}\selectfont\catcode`\^=\active\def^{\ifmmode\sp\else\^{}\fi}\catcode`\%=\active\def
\end{pgfscope}%
\end{pgfpicture}%
\makeatother%
\endgroup%

%% file: figures/optimized_parallel_n4.pgf
\begingroup%
\makeatletter%
\begin{pgfpicture}%
\pgfpathrectangle{\pgfpointorigin}{\pgfqpoint{2.077478in}{2.077478in}}%
\pgfusepath{use as bounding box, clip}%
\begin{pgfscope}%
\pgfsetbuttcap%
\pgfsetmiterjoin%
\definecolor{currentfill}{rgb}{1.000000,1.000000,1.000000}%
\pgfsetfillcolor{currentfill}%
\pgfsetlinewidth{0.000000pt}%
\definecolor{currentstroke}{rgb}{1.000000,1.000000,1.000000}%
\pgfsetstrokecolor{currentstroke}%
\pgfsetdash{}{0pt}%
\pgfpathmoveto{\pgfqpoint{0.000000in}{0.000000in}}%
\pgfpathlineto{\pgfqpoint{2.077478in}{0.000000in}}%
\pgfpathlineto{\pgfqpoint{2.077478in}{2.077478in}}%
\pgfpathlineto{\pgfqpoint{0.000000in}{2.077478in}}%
\pgfpathlineto{\pgfqpoint{0.000000in}{0.000000in}}%
\pgfpathclose%
\pgfusepath{fill}%
\end{pgfscope}%
\begin{pgfscope}%
\pgfsetbuttcap%
\pgfsetmiterjoin%
\definecolor{currentfill}{rgb}{1.000000,1.000000,1.000000}%
\pgfsetfillcolor{currentfill}%
\pgfsetlinewidth{0.000000pt}%
\definecolor{currentstroke}{rgb}{0.000000,0.000000,0.000000}%
\pgfsetstrokecolor{currentstroke}%
\pgfsetstrokeopacity{0.000000}%
\pgfsetdash{}{0pt}%
\pgfpathmoveto{\pgfqpoint{0.502394in}{0.450309in}}%
\pgfpathlineto{\pgfqpoint{1.947964in}{0.450309in}}%
\pgfpathlineto{\pgfqpoint{1.947964in}{1.938898in}}%
\pgfpathlineto{\pgfqpoint{0.502394in}{1.938898in}}%
\pgfpathlineto{\pgfqpoint{0.502394in}{0.450309in}}%
\pgfpathclose%
\pgfusepath{fill}%
\end{pgfscope}%
\begin{pgfscope}%
\pgfpathrectangle{\pgfqpoint{0.502394in}{0.450309in}}{\pgfqpoint{1.445569in}{1.488589in}}%
\pgfusepath{clip}%
\pgfsys@transformshift{0.502394in}{0.450309in}%
\pgftext[left,bottom]{\includegraphics[interpolate=true,width=1.446667in,height=1.490000in]{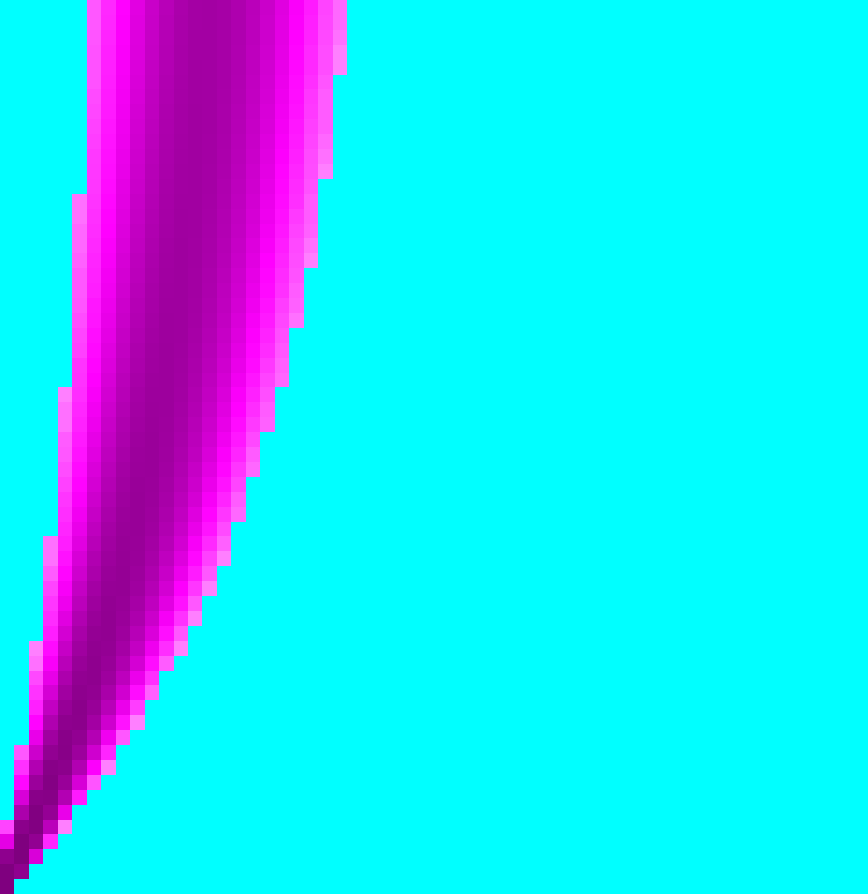}}%
\end{pgfscope}%
\begin{pgfscope}%
\pgfsetbuttcap%
\pgfsetroundjoin%
\definecolor{currentfill}{rgb}{0.000000,0.000000,0.000000}%
\pgfsetfillcolor{currentfill}%
\pgfsetlinewidth{0.803000pt}%
\definecolor{currentstroke}{rgb}{0.000000,0.000000,0.000000}%
\pgfsetstrokecolor{currentstroke}%
\pgfsetdash{}{0pt}%
\pgfsys@defobject{currentmarker}{\pgfqpoint{0.000000in}{-0.048611in}}{\pgfqpoint{0.000000in}{0.000000in}}{%
\pgfpathmoveto{\pgfqpoint{0.000000in}{0.000000in}}%
\pgfpathlineto{\pgfqpoint{0.000000in}{-0.048611in}}%
\pgfusepath{stroke,fill}%
}%
\begin{pgfscope}%
\pgfsys@transformshift{0.502394in}{0.450309in}%
\pgfsys@useobject{currentmarker}{}%
\end{pgfscope}%
\end{pgfscope}%
\begin{pgfscope}%
\definecolor{textcolor}{rgb}{0.000000,0.000000,0.000000}%
\pgfsetstrokecolor{textcolor}%
\pgfsetfillcolor{textcolor}%
\pgftext[x=0.502394in,y=0.353087in,,top]{\color{textcolor}{\rmfamily\fontsize{8.000000}{9.600000}\selectfont\catcode`\^=\active\def^{\ifmmode\sp\else\^{}\fi}\catcode`\%=\active\def
\end{pgfscope}%
\begin{pgfscope}%
\pgfsetbuttcap%
\pgfsetroundjoin%
\definecolor{currentfill}{rgb}{0.000000,0.000000,0.000000}%
\pgfsetfillcolor{currentfill}%
\pgfsetlinewidth{0.803000pt}%
\definecolor{currentstroke}{rgb}{0.000000,0.000000,0.000000}%
\pgfsetstrokecolor{currentstroke}%
\pgfsetdash{}{0pt}%
\pgfsys@defobject{currentmarker}{\pgfqpoint{0.000000in}{-0.048611in}}{\pgfqpoint{0.000000in}{0.000000in}}{%
\pgfpathmoveto{\pgfqpoint{0.000000in}{0.000000in}}%
\pgfpathlineto{\pgfqpoint{0.000000in}{-0.048611in}}%
\pgfusepath{stroke,fill}%
}%
\begin{pgfscope}%
\pgfsys@transformshift{0.984251in}{0.450309in}%
\pgfsys@useobject{currentmarker}{}%
\end{pgfscope}%
\end{pgfscope}%
\begin{pgfscope}%
\definecolor{textcolor}{rgb}{0.000000,0.000000,0.000000}%
\pgfsetstrokecolor{textcolor}%
\pgfsetfillcolor{textcolor}%
\pgftext[x=0.984251in,y=0.353087in,,top]{\color{textcolor}{\rmfamily\fontsize{8.000000}{9.600000}\selectfont\catcode`\^=\active\def^{\ifmmode\sp\else\^{}\fi}\catcode`\%=\active\def
\end{pgfscope}%
\begin{pgfscope}%
\pgfsetbuttcap%
\pgfsetroundjoin%
\definecolor{currentfill}{rgb}{0.000000,0.000000,0.000000}%
\pgfsetfillcolor{currentfill}%
\pgfsetlinewidth{0.803000pt}%
\definecolor{currentstroke}{rgb}{0.000000,0.000000,0.000000}%
\pgfsetstrokecolor{currentstroke}%
\pgfsetdash{}{0pt}%
\pgfsys@defobject{currentmarker}{\pgfqpoint{0.000000in}{-0.048611in}}{\pgfqpoint{0.000000in}{0.000000in}}{%
\pgfpathmoveto{\pgfqpoint{0.000000in}{0.000000in}}%
\pgfpathlineto{\pgfqpoint{0.000000in}{-0.048611in}}%
\pgfusepath{stroke,fill}%
}%
\begin{pgfscope}%
\pgfsys@transformshift{1.466107in}{0.450309in}%
\pgfsys@useobject{currentmarker}{}%
\end{pgfscope}%
\end{pgfscope}%
\begin{pgfscope}%
\definecolor{textcolor}{rgb}{0.000000,0.000000,0.000000}%
\pgfsetstrokecolor{textcolor}%
\pgfsetfillcolor{textcolor}%
\pgftext[x=1.466107in,y=0.353087in,,top]{\color{textcolor}{\rmfamily\fontsize{8.000000}{9.600000}\selectfont\catcode`\^=\active\def^{\ifmmode\sp\else\^{}\fi}\catcode`\%=\active\def
\end{pgfscope}%
\begin{pgfscope}%
\pgfsetbuttcap%
\pgfsetroundjoin%
\definecolor{currentfill}{rgb}{0.000000,0.000000,0.000000}%
\pgfsetfillcolor{currentfill}%
\pgfsetlinewidth{0.803000pt}%
\definecolor{currentstroke}{rgb}{0.000000,0.000000,0.000000}%
\pgfsetstrokecolor{currentstroke}%
\pgfsetdash{}{0pt}%
\pgfsys@defobject{currentmarker}{\pgfqpoint{0.000000in}{-0.048611in}}{\pgfqpoint{0.000000in}{0.000000in}}{%
\pgfpathmoveto{\pgfqpoint{0.000000in}{0.000000in}}%
\pgfpathlineto{\pgfqpoint{0.000000in}{-0.048611in}}%
\pgfusepath{stroke,fill}%
}%
\begin{pgfscope}%
\pgfsys@transformshift{1.947964in}{0.450309in}%
\pgfsys@useobject{currentmarker}{}%
\end{pgfscope}%
\end{pgfscope}%
\begin{pgfscope}%
\definecolor{textcolor}{rgb}{0.000000,0.000000,0.000000}%
\pgfsetstrokecolor{textcolor}%
\pgfsetfillcolor{textcolor}%
\pgftext[x=1.947964in,y=0.353087in,,top]{\color{textcolor}{\rmfamily\fontsize{8.000000}{9.600000}\selectfont\catcode`\^=\active\def^{\ifmmode\sp\else\^{}\fi}\catcode`\%=\active\def
\end{pgfscope}%
\begin{pgfscope}%
\definecolor{textcolor}{rgb}{0.000000,0.000000,0.000000}%
\pgfsetstrokecolor{textcolor}%
\pgfsetfillcolor{textcolor}%
\pgftext[x=1.225179in,y=0.198766in,,top]{\color{textcolor}{\rmfamily\fontsize{8.000000}{9.600000}\selectfont\catcode`\^=\active\def^{\ifmmode\sp\else\^{}\fi}\catcode`\%=\active\def
\end{pgfscope}%
\begin{pgfscope}%
\pgfsetbuttcap%
\pgfsetroundjoin%
\definecolor{currentfill}{rgb}{0.000000,0.000000,0.000000}%
\pgfsetfillcolor{currentfill}%
\pgfsetlinewidth{0.803000pt}%
\definecolor{currentstroke}{rgb}{0.000000,0.000000,0.000000}%
\pgfsetstrokecolor{currentstroke}%
\pgfsetdash{}{0pt}%
\pgfsys@defobject{currentmarker}{\pgfqpoint{-0.048611in}{0.000000in}}{\pgfqpoint{-0.000000in}{0.000000in}}{%
\pgfpathmoveto{\pgfqpoint{-0.000000in}{0.000000in}}%
\pgfpathlineto{\pgfqpoint{-0.048611in}{0.000000in}}%
\pgfusepath{stroke,fill}%
}%
\begin{pgfscope}%
\pgfsys@transformshift{0.502394in}{0.450309in}%
\pgfsys@useobject{currentmarker}{}%
\end{pgfscope}%
\end{pgfscope}%
\begin{pgfscope}%
\definecolor{textcolor}{rgb}{0.000000,0.000000,0.000000}%
\pgfsetstrokecolor{textcolor}%
\pgfsetfillcolor{textcolor}%
\pgftext[x=0.254321in, y=0.411729in, left, base]{\color{textcolor}{\rmfamily\fontsize{8.000000}{9.600000}\selectfont\catcode`\^=\active\def^{\ifmmode\sp\else\^{}\fi}\catcode`\%=\active\def
\end{pgfscope}%
\begin{pgfscope}%
\pgfsetbuttcap%
\pgfsetroundjoin%
\definecolor{currentfill}{rgb}{0.000000,0.000000,0.000000}%
\pgfsetfillcolor{currentfill}%
\pgfsetlinewidth{0.803000pt}%
\definecolor{currentstroke}{rgb}{0.000000,0.000000,0.000000}%
\pgfsetstrokecolor{currentstroke}%
\pgfsetdash{}{0pt}%
\pgfsys@defobject{currentmarker}{\pgfqpoint{-0.048611in}{0.000000in}}{\pgfqpoint{-0.000000in}{0.000000in}}{%
\pgfpathmoveto{\pgfqpoint{-0.000000in}{0.000000in}}%
\pgfpathlineto{\pgfqpoint{-0.048611in}{0.000000in}}%
\pgfusepath{stroke,fill}%
}%
\begin{pgfscope}%
\pgfsys@transformshift{0.502394in}{0.698407in}%
\pgfsys@useobject{currentmarker}{}%
\end{pgfscope}%
\end{pgfscope}%
\begin{pgfscope}%
\definecolor{textcolor}{rgb}{0.000000,0.000000,0.000000}%
\pgfsetstrokecolor{textcolor}%
\pgfsetfillcolor{textcolor}%
\pgftext[x=0.254321in, y=0.659827in, left, base]{\color{textcolor}{\rmfamily\fontsize{8.000000}{9.600000}\selectfont\catcode`\^=\active\def^{\ifmmode\sp\else\^{}\fi}\catcode`\%=\active\def
\end{pgfscope}%
\begin{pgfscope}%
\pgfsetbuttcap%
\pgfsetroundjoin%
\definecolor{currentfill}{rgb}{0.000000,0.000000,0.000000}%
\pgfsetfillcolor{currentfill}%
\pgfsetlinewidth{0.803000pt}%
\definecolor{currentstroke}{rgb}{0.000000,0.000000,0.000000}%
\pgfsetstrokecolor{currentstroke}%
\pgfsetdash{}{0pt}%
\pgfsys@defobject{currentmarker}{\pgfqpoint{-0.048611in}{0.000000in}}{\pgfqpoint{-0.000000in}{0.000000in}}{%
\pgfpathmoveto{\pgfqpoint{-0.000000in}{0.000000in}}%
\pgfpathlineto{\pgfqpoint{-0.048611in}{0.000000in}}%
\pgfusepath{stroke,fill}%
}%
\begin{pgfscope}%
\pgfsys@transformshift{0.502394in}{0.946505in}%
\pgfsys@useobject{currentmarker}{}%
\end{pgfscope}%
\end{pgfscope}%
\begin{pgfscope}%
\definecolor{textcolor}{rgb}{0.000000,0.000000,0.000000}%
\pgfsetstrokecolor{textcolor}%
\pgfsetfillcolor{textcolor}%
\pgftext[x=0.254321in, y=0.907925in, left, base]{\color{textcolor}{\rmfamily\fontsize{8.000000}{9.600000}\selectfont\catcode`\^=\active\def^{\ifmmode\sp\else\^{}\fi}\catcode`\%=\active\def
\end{pgfscope}%
\begin{pgfscope}%
\pgfsetbuttcap%
\pgfsetroundjoin%
\definecolor{currentfill}{rgb}{0.000000,0.000000,0.000000}%
\pgfsetfillcolor{currentfill}%
\pgfsetlinewidth{0.803000pt}%
\definecolor{currentstroke}{rgb}{0.000000,0.000000,0.000000}%
\pgfsetstrokecolor{currentstroke}%
\pgfsetdash{}{0pt}%
\pgfsys@defobject{currentmarker}{\pgfqpoint{-0.048611in}{0.000000in}}{\pgfqpoint{-0.000000in}{0.000000in}}{%
\pgfpathmoveto{\pgfqpoint{-0.000000in}{0.000000in}}%
\pgfpathlineto{\pgfqpoint{-0.048611in}{0.000000in}}%
\pgfusepath{stroke,fill}%
}%
\begin{pgfscope}%
\pgfsys@transformshift{0.502394in}{1.194603in}%
\pgfsys@useobject{currentmarker}{}%
\end{pgfscope}%
\end{pgfscope}%
\begin{pgfscope}%
\definecolor{textcolor}{rgb}{0.000000,0.000000,0.000000}%
\pgfsetstrokecolor{textcolor}%
\pgfsetfillcolor{textcolor}%
\pgftext[x=0.254321in, y=1.156023in, left, base]{\color{textcolor}{\rmfamily\fontsize{8.000000}{9.600000}\selectfont\catcode`\^=\active\def^{\ifmmode\sp\else\^{}\fi}\catcode`\%=\active\def
\end{pgfscope}%
\begin{pgfscope}%
\pgfsetbuttcap%
\pgfsetroundjoin%
\definecolor{currentfill}{rgb}{0.000000,0.000000,0.000000}%
\pgfsetfillcolor{currentfill}%
\pgfsetlinewidth{0.803000pt}%
\definecolor{currentstroke}{rgb}{0.000000,0.000000,0.000000}%
\pgfsetstrokecolor{currentstroke}%
\pgfsetdash{}{0pt}%
\pgfsys@defobject{currentmarker}{\pgfqpoint{-0.048611in}{0.000000in}}{\pgfqpoint{-0.000000in}{0.000000in}}{%
\pgfpathmoveto{\pgfqpoint{-0.000000in}{0.000000in}}%
\pgfpathlineto{\pgfqpoint{-0.048611in}{0.000000in}}%
\pgfusepath{stroke,fill}%
}%
\begin{pgfscope}%
\pgfsys@transformshift{0.502394in}{1.442701in}%
\pgfsys@useobject{currentmarker}{}%
\end{pgfscope}%
\end{pgfscope}%
\begin{pgfscope}%
\definecolor{textcolor}{rgb}{0.000000,0.000000,0.000000}%
\pgfsetstrokecolor{textcolor}%
\pgfsetfillcolor{textcolor}%
\pgftext[x=0.254321in, y=1.404121in, left, base]{\color{textcolor}{\rmfamily\fontsize{8.000000}{9.600000}\selectfont\catcode`\^=\active\def^{\ifmmode\sp\else\^{}\fi}\catcode`\%=\active\def
\end{pgfscope}%
\begin{pgfscope}%
\pgfsetbuttcap%
\pgfsetroundjoin%
\definecolor{currentfill}{rgb}{0.000000,0.000000,0.000000}%
\pgfsetfillcolor{currentfill}%
\pgfsetlinewidth{0.803000pt}%
\definecolor{currentstroke}{rgb}{0.000000,0.000000,0.000000}%
\pgfsetstrokecolor{currentstroke}%
\pgfsetdash{}{0pt}%
\pgfsys@defobject{currentmarker}{\pgfqpoint{-0.048611in}{0.000000in}}{\pgfqpoint{-0.000000in}{0.000000in}}{%
\pgfpathmoveto{\pgfqpoint{-0.000000in}{0.000000in}}%
\pgfpathlineto{\pgfqpoint{-0.048611in}{0.000000in}}%
\pgfusepath{stroke,fill}%
}%
\begin{pgfscope}%
\pgfsys@transformshift{0.502394in}{1.690800in}%
\pgfsys@useobject{currentmarker}{}%
\end{pgfscope}%
\end{pgfscope}%
\begin{pgfscope}%
\definecolor{textcolor}{rgb}{0.000000,0.000000,0.000000}%
\pgfsetstrokecolor{textcolor}%
\pgfsetfillcolor{textcolor}%
\pgftext[x=0.254321in, y=1.652219in, left, base]{\color{textcolor}{\rmfamily\fontsize{8.000000}{9.600000}\selectfont\catcode`\^=\active\def^{\ifmmode\sp\else\^{}\fi}\catcode`\%=\active\def
\end{pgfscope}%
\begin{pgfscope}%
\pgfsetbuttcap%
\pgfsetroundjoin%
\definecolor{currentfill}{rgb}{0.000000,0.000000,0.000000}%
\pgfsetfillcolor{currentfill}%
\pgfsetlinewidth{0.803000pt}%
\definecolor{currentstroke}{rgb}{0.000000,0.000000,0.000000}%
\pgfsetstrokecolor{currentstroke}%
\pgfsetdash{}{0pt}%
\pgfsys@defobject{currentmarker}{\pgfqpoint{-0.048611in}{0.000000in}}{\pgfqpoint{-0.000000in}{0.000000in}}{%
\pgfpathmoveto{\pgfqpoint{-0.000000in}{0.000000in}}%
\pgfpathlineto{\pgfqpoint{-0.048611in}{0.000000in}}%
\pgfusepath{stroke,fill}%
}%
\begin{pgfscope}%
\pgfsys@transformshift{0.502394in}{1.938898in}%
\pgfsys@useobject{currentmarker}{}%
\end{pgfscope}%
\end{pgfscope}%
\begin{pgfscope}%
\definecolor{textcolor}{rgb}{0.000000,0.000000,0.000000}%
\pgfsetstrokecolor{textcolor}%
\pgfsetfillcolor{textcolor}%
\pgftext[x=0.254321in, y=1.900317in, left, base]{\color{textcolor}{\rmfamily\fontsize{8.000000}{9.600000}\selectfont\catcode`\^=\active\def^{\ifmmode\sp\else\^{}\fi}\catcode`\%=\active\def
\end{pgfscope}%
\begin{pgfscope}%
\definecolor{textcolor}{rgb}{0.000000,0.000000,0.000000}%
\pgfsetstrokecolor{textcolor}%
\pgfsetfillcolor{textcolor}%
\pgftext[x=0.198766in,y=1.194603in,,bottom,rotate=90.000000]{\color{textcolor}{\rmfamily\fontsize{8.000000}{9.600000}\selectfont\catcode`\^=\active\def^{\ifmmode\sp\else\^{}\fi}\catcode`\%=\active\def
\end{pgfscope}%
\begin{pgfscope}%
\pgfpathrectangle{\pgfqpoint{0.502394in}{0.450309in}}{\pgfqpoint{1.445569in}{1.488589in}}%
\pgfusepath{clip}%
\pgfsetbuttcap%
\pgfsetroundjoin%
\pgfsetlinewidth{1.003750pt}%
\definecolor{currentstroke}{rgb}{0.000000,0.000000,0.000000}%
\pgfsetstrokecolor{currentstroke}%
\pgfsetdash{{3.700000pt}{1.600000pt}}{0.000000pt}%
\pgfpathmoveto{\pgfqpoint{0.502394in}{0.450309in}}%
\pgfpathlineto{\pgfqpoint{1.947964in}{1.938898in}}%
\pgfusepath{stroke}%
\end{pgfscope}%
\begin{pgfscope}%
\pgfsetrectcap%
\pgfsetmiterjoin%
\pgfsetlinewidth{0.803000pt}%
\definecolor{currentstroke}{rgb}{0.000000,0.000000,0.000000}%
\pgfsetstrokecolor{currentstroke}%
\pgfsetdash{}{0pt}%
\pgfpathmoveto{\pgfqpoint{0.502394in}{0.450309in}}%
\pgfpathlineto{\pgfqpoint{0.502394in}{1.938898in}}%
\pgfusepath{stroke}%
\end{pgfscope}%
\begin{pgfscope}%
\pgfsetrectcap%
\pgfsetmiterjoin%
\pgfsetlinewidth{0.803000pt}%
\definecolor{currentstroke}{rgb}{0.000000,0.000000,0.000000}%
\pgfsetstrokecolor{currentstroke}%
\pgfsetdash{}{0pt}%
\pgfpathmoveto{\pgfqpoint{1.947964in}{0.450309in}}%
\pgfpathlineto{\pgfqpoint{1.947964in}{1.938898in}}%
\pgfusepath{stroke}%
\end{pgfscope}%
\begin{pgfscope}%
\pgfsetrectcap%
\pgfsetmiterjoin%
\pgfsetlinewidth{0.803000pt}%
\definecolor{currentstroke}{rgb}{0.000000,0.000000,0.000000}%
\pgfsetstrokecolor{currentstroke}%
\pgfsetdash{}{0pt}%
\pgfpathmoveto{\pgfqpoint{0.502394in}{0.450309in}}%
\pgfpathlineto{\pgfqpoint{1.947964in}{0.450309in}}%
\pgfusepath{stroke}%
\end{pgfscope}%
\begin{pgfscope}%
\pgfsetrectcap%
\pgfsetmiterjoin%
\pgfsetlinewidth{0.803000pt}%
\definecolor{currentstroke}{rgb}{0.000000,0.000000,0.000000}%
\pgfsetstrokecolor{currentstroke}%
\pgfsetdash{}{0pt}%
\pgfpathmoveto{\pgfqpoint{0.502394in}{1.938898in}}%
\pgfpathlineto{\pgfqpoint{1.947964in}{1.938898in}}%
\pgfusepath{stroke}%
\end{pgfscope}%
\begin{pgfscope}%
\pgfsetbuttcap%
\pgfsetmiterjoin%
\definecolor{currentfill}{rgb}{1.000000,1.000000,1.000000}%
\pgfsetfillcolor{currentfill}%
\pgfsetfillopacity{0.800000}%
\pgfsetlinewidth{1.003750pt}%
\definecolor{currentstroke}{rgb}{0.800000,0.800000,0.800000}%
\pgfsetstrokecolor{currentstroke}%
\pgfsetstrokeopacity{0.800000}%
\pgfsetdash{}{0pt}%
\pgfpathmoveto{\pgfqpoint{1.232605in}{0.505864in}}%
\pgfpathlineto{\pgfqpoint{1.870186in}{0.505864in}}%
\pgfpathquadraticcurveto{\pgfqpoint{1.892408in}{0.505864in}}{\pgfqpoint{1.892408in}{0.528087in}}%
\pgfpathlineto{\pgfqpoint{1.892408in}{0.671914in}}%
\pgfpathquadraticcurveto{\pgfqpoint{1.892408in}{0.694136in}}{\pgfqpoint{1.870186in}{0.694136in}}%
\pgfpathlineto{\pgfqpoint{1.232605in}{0.694136in}}%
\pgfpathquadraticcurveto{\pgfqpoint{1.210382in}{0.694136in}}{\pgfqpoint{1.210382in}{0.671914in}}%
\pgfpathlineto{\pgfqpoint{1.210382in}{0.528087in}}%
\pgfpathquadraticcurveto{\pgfqpoint{1.210382in}{0.505864in}}{\pgfqpoint{1.232605in}{0.505864in}}%
\pgfpathlineto{\pgfqpoint{1.232605in}{0.505864in}}%
\pgfpathclose%
\pgfusepath{stroke,fill}%
\end{pgfscope}%
\begin{pgfscope}%
\pgfsetbuttcap%
\pgfsetroundjoin%
\pgfsetlinewidth{1.003750pt}%
\definecolor{currentstroke}{rgb}{0.000000,0.000000,0.000000}%
\pgfsetstrokecolor{currentstroke}%
\pgfsetdash{{3.700000pt}{1.600000pt}}{0.000000pt}%
\pgfpathmoveto{\pgfqpoint{1.254827in}{0.610803in}}%
\pgfpathlineto{\pgfqpoint{1.365938in}{0.610803in}}%
\pgfpathlineto{\pgfqpoint{1.477049in}{0.610803in}}%
\pgfusepath{stroke}%
\end{pgfscope}%
\begin{pgfscope}%
\definecolor{textcolor}{rgb}{0.000000,0.000000,0.000000}%
\pgfsetstrokecolor{textcolor}%
\pgfsetfillcolor{textcolor}%
\pgftext[x=1.565938in,y=0.571914in,left,base]{\color{textcolor}{\rmfamily\fontsize{8.000000}{9.600000}\selectfont\catcode`\^=\active\def^{\ifmmode\sp\else\^{}\fi}\catcode`\%=\active\def
\end{pgfscope}%
\end{pgfpicture}%
\makeatother%
\endgroup%